\documentclass[aoas,preprint]{imsart}
\RequirePackage{amsthm,amsmath,amsfonts,amssymb}
\RequirePackage[authoryear]{natbib}
\RequirePackage[colorlinks,citecolor=blue,urlcolor=blue]{hyperref}
\RequirePackage{graphicx}

\startlocaldefs

\theoremstyle{plain}
\newtheorem{theorem}{Theorem}

\newtheorem{lemma}[theorem]{Lemma}

\theoremstyle{definition}

\newtheorem{example}{Example}
\newtheorem{remark}{Remark}

\usepackage{mathtools}
\mathtoolsset{showonlyrefs=true}

\usepackage{proba}
\usepackage{bm}
\usepackage{url}
\usepackage{color}
\usepackage{adjustbox}
\usepackage{anyfontsize}

\def\theenumi{\alph{enumi}}

\DeclareMathOperator*{\Exp}{\E}
\newcommand{\PMSE}{\operatorname{PMSE}}
\newcommand{\MISE}{\operatorname{MISE}}
\newcommand{\AMISE}{\operatorname{AMISE}}

\graphicspath{{figs/}}

\endlocaldefs

\begin{document}

\begin{frontmatter}

\title{Fast boundary-aware spatial intensity estimation on complex domains}
\runtitle{Boundary-aware spatial intensity estimation}
\pdfsubject{}

\begin{aug}
\author[A]{\fnms{Takumi}~\snm{Nakagawa}\ead[label=e1]{nakagawa.t.as@m.titech.ac.jp}}
\author[B]{\fnms{K\={o}saku}~\snm{Takanashi}\ead[label=e2]{kosaku.takanashi@riken.jp}}
\author[C]{\fnms{Kenichiro}~\snm{McAlinn}\ead[label=e3]{kenichiro.mcalinn@temple.edu}}
\author[A,B]{\fnms{Takafumi}~\snm{Kanamori}\ead[label=e4]{kanamori@comp.isct.ac.jp}}

\address[A]{Department of Mathematical and Computing Science,
Institute of Science Tokyo\printead[presep={,\ \\ }]{e1,e4}}

\address[B]{Center for Advanced Intelligence Project,
RIKEN\printead[presep={,\ \\ }]{e2}}

\address[C]{Department of Statistics, Operations, and Data Science,
Temple University\printead[presep={,\ \\ }]{e3}}
\end{aug}

\begin{abstract}
Spatial intensity maps are routinely used to summarize point patterns on geographically constrained regions, such as islands, coastlines, watersheds, ecological reserves, and administrative areas with physical barriers. In these settings, the domain is not a nuisance feature-- it determines where probability mass may be assigned and which locations should be smoothed together. Standard kernel density estimators can place mass outside the study region and smooth according to Euclidean distance, while diffusion-based estimators respect the geometry but are costly to recompute when the data subset, bandwidth, or evaluation grid changes. Motivated by repeated hotspot mapping of theft and larceny incidents on Oahu, Hawaii, we propose the projected diffusion kernel density estimator (PDKDE). PDKDE expands the Neumann diffusion kernel in a truncated Laplacian eigenbasis, so that the geometry of a fixed domain is computed once and subsequent density estimates are obtained through explicit spectral coefficients. The resulting estimator preserves the boundary-aware and geometry-respecting behavior of diffusion smoothing while making repeated estimation and least-squares cross-validation computationally practical. For the exact projected estimator, we prove MISE consistency and pointwise consistency up to the boundary. Controlled simulations show that PDKDE reduces boundary and barrier artifacts relative to Euclidean kernels and is substantially faster than the geometry-aware comparators considered after the one-time domain computation. In the Oahu application, the method produces coastline-constrained descriptive maps across time windows in seconds, illustrating its intended fixed-domain, changing-data use case.
\end{abstract}


\end{frontmatter}

\hypersetup{pdfsubject={}}

\section{Introduction}\label{sec:intro}

Spatial point pattern analyses often begin with a smoothed intensity map. Such maps are used to summarize where events are concentrated, compare patterns across time periods or subpopulations, and communicate spatial structure to domain experts. In many applications, however, the relevant support is not all of $\mathbb{R}^d$ but a bounded and geometrically complex region. The shape of this region is part of the statistical problem: probability mass should not be assigned to infeasible locations, and smoothing should not pass through physical barriers or across separated parts of the domain.

Our motivating application is the construction of spatial intensity maps for theft and larceny incidents on Oahu, Hawaii. As illustrated in Figure~\ref{hawaii_data}, many incidents occur near the coastline, while the ocean is outside the support of the process being mapped. A useful estimator in this setting must therefore do more than produce a visually smooth surface. It must assign density only on the island, preserve intensity along the coastline rather than pulling it artificially inland, adapt to the island's geometry, and remain fast enough to be recomputed across time windows and data updates. Similar requirements arise in ecological surveys on protected regions, environmental monitoring in lakes or watersheds, disease mapping on administrative areas, and other spatial analyses where the study domain is a scientific constraint rather than a plotting boundary.

\begin{figure}[t]
    \centering \includegraphics[width=0.7\textwidth]{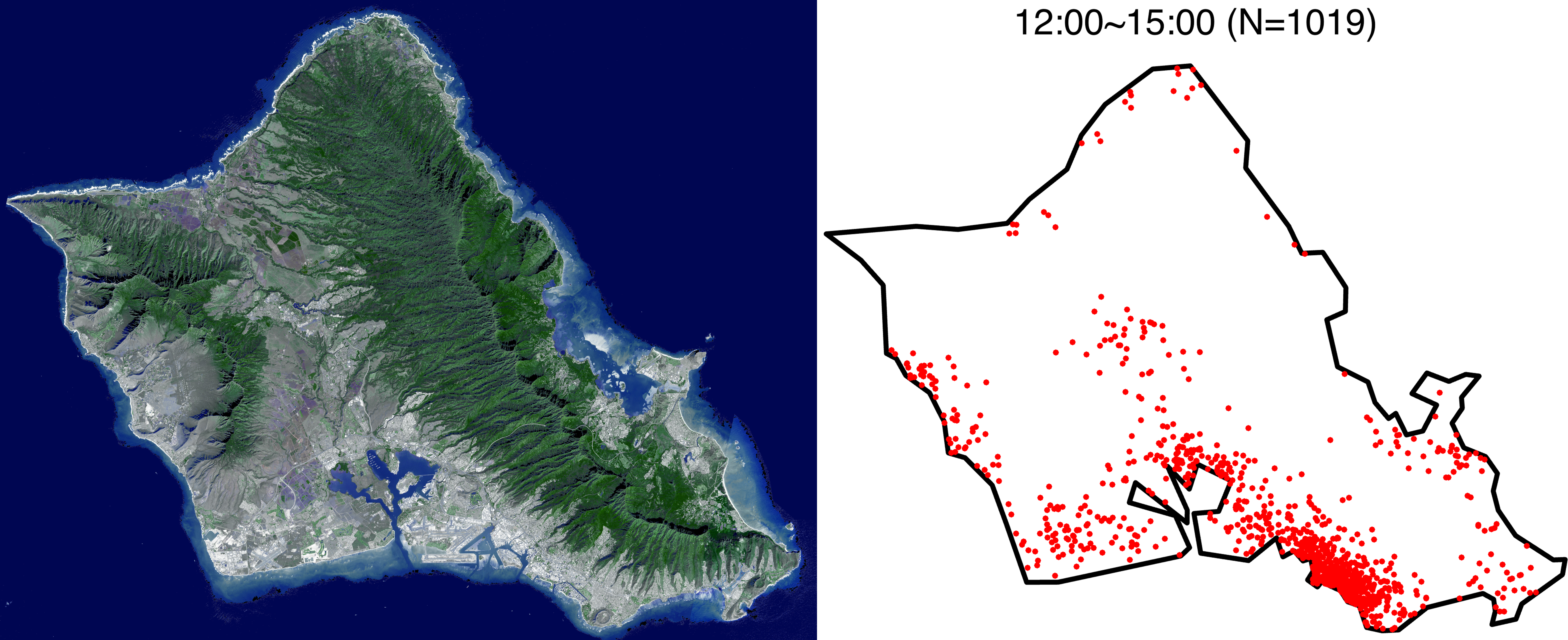}
    \caption{(Left) Satellite image of Oahu, Hawaii (Image Credit: NASA/GSFC/METI/ERSDAC/JAROS, and U.S./Japan ASTER Science Team) and (Right) theft/larceny data from 11/27/2024  to 5/28/2025 between 12:00 and 15:00.
    }\label{hawaii_data}
\end{figure}

Existing methods capture only part of this applied need. Boundary-corrected kernels such as renormalization \citep{jones1993simple} reduce leakage outside the domain, but they still smooth according to Euclidean distance and can therefore blur across barriers or non-convex geometry. Diffusion-based KDE \citep{botev2010kernel} is more faithful to the support because smoothing is defined by heat flow inside the domain with an insulated boundary. Its direct use, however, requires repeatedly solving a partial differential equation (PDE), making bandwidth selection and repeated estimation costly on realistic spatial domains \citep{baddeley2022diffusion}.

We propose the projected diffusion kernel density estimator (PDKDE) as a practical way to address this gap. The estimator starts from the Neumann diffusion kernel, which respects the boundary and geometry of the domain, and then represents that kernel through a truncated Laplacian eigenbasis. This converts diffusion smoothing from a repeated PDE computation into an explicit spectral estimator. The domain-specific eigenbasis is computed once for a fixed study region; after that, new estimates, evaluation grids, and bandwidth choices require only projection coefficients and analytic expressions. The method is therefore designed for the fixed-domain, changing-data setting that arises in time-binned hotspot maps and related spatial monitoring problems.

This paper makes five main contributions:
\begin{enumerate}
\item We formulate boundary-aware KDE for complex spatial domains as a fixed-domain, changing-data problem motivated by real geographic intensity mapping.
\item We introduce PDKDE, a spectral approximation to diffusion KDE that preserves boundary and geometry adaptation while avoiding repeated PDE solves.
\item We derive an analytic least-squares cross-validation criterion, making automatic bandwidth selection feasible without manual tuning or repeated numerical PDE computation.
\item We provide theoretical guarantees showing that the exact projected estimator is MISE-consistent and pointwise consistent up to the boundary when the truncation level grows appropriately with the sample size and bandwidth.
\item Through controlled geometric stress tests and an application to Oahu theft and larceny data, we show that PDKDE reduces visible boundary and geometry artifacts while providing the speed needed for repeated spatial analyses on a fixed complex domain.
\end{enumerate}

Table~\ref{table:comparison} summarizes the gap addressed by PDKDE. The comparison is qualitative, because some methods are excellent on simple supports but do not extend cleanly to arbitrary coastlines, holes, and barriers. The relevant distinction for our setting is whether a method treats the domain as the statistical support and can be reused efficiently when the same support is analyzed repeatedly.

The remainder of this paper is organized as follows. Section~\ref{sec:related} reviews related literature. Section~\ref{sec:dkde} introduces the basic concepts of diffusion KDE and motivates our projected construction. Section~\ref{sec:pdkde} develops our proposed PDKDE, detailing its spectral derivation, finite-element implementation, and an analytic least-squares cross-validation criterion for efficient bandwidth selection. Section~\ref{sec:experiments} presents controlled numerical experiments, while Section~\ref{sec:application} discusses the Oahu theft and larceny application. Section~\ref{sec:discussion} concludes the paper. All proofs and additional experiments are provided in the appendices.

\begin{table}[t!]
    \centering
    \caption{Qualitative comparison of density estimators for bounded spatial domains. ``Support'' means that estimates are defined as densities on the study region rather than on an artificial Euclidean extension. ``Barrier aware'' means that smoothing respects non-convex geometry, holes, or physical barriers.}\label{table:comparison}
    \begin{adjustbox}{width=\textwidth}
    \begin{tabular}{lllllll}
    \hline
    Method & Support & Boundary behavior & Barrier aware & Repeated use on fixed $D$ & Bandwidth selection & Implementation burden \\ \hline
    Standard KDE & No & Biased near boundary & No & Fast but invalid & Standard plug-in/CV & Low \\
    Reflection & Simple domains & Good on flat/product boundaries & No & Fast on simple domains & Standard plug-in/CV & Low--moderate \\
    Renormalized KDE & Yes & Reduces leakage & No; Euclidean smoothing & Slow if normalizers recomputed & Standard plug-in/CV & Moderate \\
    Direct diffusion KDE & Yes & Neumann boundary & Yes & Slow; repeated PDE solves & Costly CV & High \\
    DE-PDE / \texttt{fdaPDE} & Yes & Empirically good & Yes & Slow if refit repeatedly & CV over penalty grid & High \\
    PDKDE & Yes & Pointwise consistent up to boundary$^{\dagger}$ & Yes & Fast after eigenbasis & Analytic LSCV & Moderate after preprocessing \\ \hline
    \end{tabular}
    \end{adjustbox}
    \begin{flushleft}\footnotesize
    $^{\dagger}$ The formal consistency statement applies to the exact truncated spectral estimator $g_M$ under the assumptions of Theorems~\ref{proposition: MISE of g_M}--\ref{proposition: g_M, weak consistency}; the finite-element approximation adds the discretization error discussed in Section~\ref{sec:pdkde}.
    \end{flushleft}
\end{table}

\paragraph*{Notation}

For an open set $D\subset \mathbb{R}^d$, $\overline D$ is its closure, $\partial D = \overline D \setminus D$ its boundary, and $|D| = \int_D dx$ its Lebesgue measure. The Laplacian is $\Delta_x = \sum_{i=1}^d \frac{\partial^2}{\partial x_i^2}$ (or simply $\Delta$), and $\delta$ is the Dirac delta function. $C(D)$ is the set of continuous functions on $D$, and $C^2(D)$ is the set of twice continuously differentiable functions on $D$. $L^2(D)$ is the space of square-integrable functions with inner product $\langle f, g \rangle_{L^2(D)} = \int_D f(x) g(x) dx$ and norm $\|f\|_{L^2(D)} = \sqrt{\langle f, f \rangle_{L^2(D)}}$. We define the supremum norm as $\|f\|_\infty := \sup_{x\in D} |f(x)|$. For a random variable $X$, $\Exp [X]$ and $\mathrm{Var}[X]$ denote its expectation and variance. For a density estimator $\hat{f}$ of $f$, its mean integrated squared error (MISE) is defined as $\MISE(\hat{f}) \coloneqq \Exp\|\hat{f}-f\|_{L^{2}(D)}^{2}$. For functions $a, b: \mathbb{R} \to \mathbb{R}$, we write $a(n) = \Theta(b(n))$ if there exist positive constants $C_1, C_2$ and $n_0$ such that $C_1 |b(n)| \leq |a(n)| \leq C_2 |b(n)|$ for all $n > n_0$. We write $a(n) = \omega(b(n))$ if for any positive constant $C$, there exists $n_0$ such that $|a(n)| \geq C |b(n)|$ for all $n > n_0$.

\section{Related Work}\label{sec:related}

This work lies at the intersection of boundary-corrected density estimation, diffusion-based smoothing, and computational methods for spatial data on complex supports. The common applied question across these literatures is whether the resulting intensity surface is a credible description of events on the actual study region, rather than on an artificial Euclidean extension of that region.

Kernel density estimation with symmetric kernels is known to suffer from boundary bias when the support has a boundary. In one dimension, many correction strategies have been proposed, including reflection \citep{jones1993simple,KARUNAMUNI2008some,zhang1999improved}, renormalization, generalized jackknife methods \citep{jones1993simple}, and transformations \citep{marron1994transformations,wand1991transformations}. These methods are effective for simple supports. Reflection, for example, can be used on intervals and simple product domains, but does not provide a general construction for arbitrary coastlines, holes, or barriers. The applied spatial settings motivating this paper are multidimensional and irregular in exactly this sense.

Renormalization and generalized jackknife methods can, in principle, be applied to multidimensional domains \citep{MARSHALL2010boundary,HAZELTON2009linear}. Their practical use is limited by numerical integration over the domain, and their smoothing remains Euclidean. Consequently, they can correct the amount of mass near a boundary without correcting which parts of the domain should communicate through smoothing. This distinction matters on domains with narrow channels, holes, or strong barriers \citep{ferraccioli2021nonparametric,baddeley2022diffusion}.

Asymmetric kernels provide another route to boundary correction. Examples include the gamma kernel \citep{chen2000probability} and the diffusion kernel \citep{botev2010kernel}. The diffusion kernel is especially appealing for spatial analysis because it is a natural extension of the Gaussian kernel to bounded domains and respects the geometry of the support through the heat equation \citep{ferraccioli2021nonparametric,baddeley2022diffusion}. It can also be interpreted as a generalization of reflection to complex boundaries \citep{botev2010kernel,baddeley2022diffusion}.

Several papers study the statistical properties of diffusion-based estimators. In one dimension, \cite{botev2010kernel} derived asymptotic MISE expressions and established boundary consistency. Related analyses include \cite{colbrook2020kernel} for linked boundary conditions and \cite{bates2014density} for manifolds without boundary. What remains less developed is a multidimensional treatment that is both theoretically supported and computationally usable on arbitrary spatial domains.

The computational barrier is the main reason diffusion KDE has not become a routine option for applied spatial intensity estimation. Existing implementations often solve the diffusion equation numerically, for example by Euler-type schemes \citep{ferraccioli2021nonparametric,baddeley2022diffusion,pelz2023diffusion}. Because the computation must be repeated when the data or smoothing parameter changes, cross-validation and time-window analyses can become expensive. Analytical series representations are available on simple domains such as rectangles \citep{botev2010kernel,yang2023exact}, but not directly on general geographic shapes.

Our contribution is to approximate the Neumann diffusion estimator by a finite spectral estimator that retains the relevant support and geometry behavior, admits MISE and pointwise consistency results up to the boundary for the exact projected estimator, and replaces repeated PDE solves by reusable domain-specific eigenfunctions.

\section{Background: Diffusion Kernel Density Estimator}\label{sec:dkde}

This section introduces diffusion KDE as the boundary-aware target estimator that PDKDE approximates. The emphasis is on the two features that matter in the motivating spatial applications: mass is retained inside the study region, and smoothing follows the geometry of that region. Detailed proofs and extended theoretical derivations are provided in the appendices.

Let $D\subset \mathbb{R}^d$ be a bounded open set with piecewise smooth boundary, and let $\overline D$ be its closure. We observe i.i.d. data $X_1,\dotsc,X_N\sim f$, where $f:\overline D\to\mathbb{R}$ is a continuous probability density on $\overline D$. A kernel density estimator has the form
\begin{equation}\label{eq: Gaussian kernel density estimator}
\hat{f}(x;t)=\frac{1}{N}\sum_{i=1}^{N}k(x;X_{i};t),\qquad x\in\overline{D},\ t>0,
\end{equation}
where $k$ is the kernel and $\sqrt{t}$ plays the role of bandwidth.

The classical Gaussian kernel is
\begin{equation}
\phi(x;y;t)=\frac{1}{\sqrt{(2\pi t)^{d}}}e^{-\frac{|x-y|^{2}}{2t}},\qquad x,y\in\mathbb{R}^{d},\ t>0.
\end{equation}
It assigns more weight to observations that are close to $x$ in Euclidean distance. This works well on $\mathbb{R}^d$, but when the support is bounded, it assigns nonzero probability mass outside the domain $D$, introducing severe boundary bias.

A common boundary correction method is the renormalized Gaussian kernel
\begin{equation}
\phi_{\mathrm{R}}(x;y;t)=\frac{\phi(x;y;t)}{\int_{D}\phi(x;z;t)dz},\qquad x\in\overline{D},\ t>0,\label{eq: renormalization kernel}
\end{equation}
which adjusts the estimation near the boundary and can reduce boundary bias \citep{jones1993simple, MARSHALL2010boundary}. Even so, renormalization still smooths according to Euclidean distance and may fail to capture the intrinsic geometry of a complex domain.

Figure~\ref{fig: kernel visualization} illustrates this issue on a U-shaped domain. Two points may be close in Euclidean distance while being separated by the geometry of the support. Gaussian and renormalized kernels do not distinguish these cases, whereas the diffusion kernel does.

\begin{figure}[t!]
\centering
\includegraphics[width=0.98\linewidth]{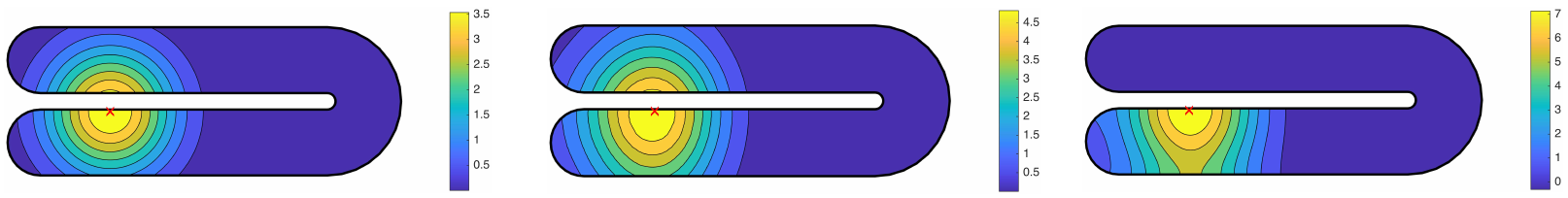}
\caption{Kernels on a U-shaped domain (red x indicates the kernel center): Gaussian kernel (left), renormalized Gaussian kernel (center), and diffusion kernel (right), all with the same bandwidth.}
\label{fig: kernel visualization}
\end{figure}

A natural way to enforce both boundary correction and geometric adaptation is to reinterpret smoothing as a heat diffusion process within the domain. The kernel then represents the heat distribution that gradually diffuses from a data point and is insulated at the boundary. This leads to the heat equation with Neumann boundary conditions.

The diffusion kernel density estimator, or simply the diffusion estimator, is defined as
\begin{equation}
g(x;t)=\frac{1}{N}\sum_{i=1}^{N}\kappa(x;X_{i};t),\qquad x\in\overline{D},\ t>0,\label{eq: diffusion estimator}
\end{equation}
where the diffusion kernel $\kappa(x;y;t)$ is the fundamental solution of
\begin{equation}
    \left\{\begin{aligned}
        &\frac{1}{2} \Delta_x \kappa(x;y;t) = \frac{\partial}{\partial t} \kappa(x;y;t) & &(x\in D,\ t > 0),\\
        &\frac{\partial \kappa}{\partial \bm n} = 0 & &(x\in\partial D,\ t > 0),\\
        &\kappa(x;y;0) = \delta(x-y) & &(x\in D).
    \end{aligned}\right. \label{eq: diffusion kernel, pde}
\end{equation}
Here, $\bm n$ is the outward unit normal vector on $\partial D$.

Although it is defined through a PDE, the diffusion estimator is still a kernel density estimator in the sense of \eqref{eq: Gaussian kernel density estimator}: it averages one kernel contribution per observation. It also recovers familiar special cases. If $D=\mathbb{R}^d$, then $\kappa(x;y;t)=\phi(x;y;t)$. If $D=(0,\infty)\subset\mathbb{R}$, then the solution becomes $\kappa(x;y;t)=\phi(x;y;t)+\phi(x;-y;t)$, which is precisely the reflection method \citep{silverman1986density}.

The diffusion estimator has two properties that are especially important for our purposes. First, by interpreting the smoothing process as heat diffusion, the kernel naturally adapts to the geometric structure of the domain rather than relying on Euclidean distance \citep{gyrya2011Neumann}. Second, because the boundary is insulated, the probability mass does not leak outside the domain, substantially mitigating boundary bias. Detailed theoretical properties of the diffusion estimator are provided in the appendices.

These properties make diffusion KDE a natural target for density estimation on irregular supports. For applied use, however, the estimator must also be recomputed across bandwidths, time windows, or updated datasets. Direct implementations require repeated PDE solves, so the statistically appropriate target is often too slow to use routinely. This computational bottleneck motivates the projected approximation developed next.

\section{Projected Diffusion Kernel Density Estimator}\label{sec:pdkde}

The main obstacle to using diffusion KDE directly is computational. In the spatial settings motivating this paper, the domain is fixed but the analyst may need estimates for different time periods, event categories, bandwidths, or evaluation grids. Generic PDE-based implementations must be rerun when these inputs change. The goal of this section is to recast the diffusion estimator in a spectral form that preserves its statistical meaning while making repeated fixed-domain estimation computationally practical.

When $D$ is bounded, the diffusion kernel admits the eigenfunction expansion \citep{ito1957fundamental}
\begin{equation}
    \kappa(x;y;t) = \sum_{m=1}^\infty e^{-\frac{1}{2} \lambda_m t} u_m(x) u_m(y). \label{eq: expanded diffusion kernel}
\end{equation}
The series converges absolutely and uniformly on $\overline{D}\times \overline{D}$ for any $t>0$. Here $\lambda_m$ and $u_m$ are the eigenvalues and eigenfunctions of the Neumann Laplacian,
\begin{equation}
    \left\{\begin{aligned}
        & - \Delta u = \lambda u & &(x\in D),\\
        &\frac{\partial u}{\partial \bm n} = 0 & &(x\in\partial D),
    \end{aligned}\right. \label{eq: eigenvalue, eigenfunction}
\end{equation}
with $\lambda_1\leq\lambda_2\leq\cdots$, and $\{u_m\}_{m=1}^\infty$ forming a complete orthonormal system of $L^2(D)$.

This expansion converts the diffusion PDE into a domain-specific eigenvalue problem. Crucially, the eigenpairs depend only on $D$, not on the observed data. Once the eigenbasis has been computed for a fixed domain, new datasets, subsamples, and bandwidths can be handled by updating only low-dimensional projection coefficients.

Substituting \eqref{eq: expanded diffusion kernel} into \eqref{eq: diffusion estimator} gives
\begin{align}
    g(x;t) &= \frac{1}{N} \sum_{i=1}^N \sum_{m=1}^\infty e^{-\frac{1}{2} \lambda_m t} u_m(x) u_m(X_i) \\
    &= \sum_{m=1}^\infty e^{-\frac{1}{2} \lambda_m t} u_m(x) \frac{1}{N} \sum_{i=1}^N u_m(X_i).
\end{align}
Approximating this by the first $M$ terms yields
\begin{equation}
    g_M(x;t) = \sum_{m=1}^M e^{-\frac{1}{2} \lambda_m t} u_m(x) \frac{1}{N} \sum_{i=1}^N u_m(X_i).\label{eq: diffusion estimator, partial sum}
\end{equation}
Because $g_M(\cdot;t)$ is the projection of $g(\cdot;t)$ onto $\mathrm{span}\{u_1,\dotsc,u_M\}\subset L^2(D)$, we call it the \emph{projected diffusion kernel density estimator} (PDKDE), or simply the \emph{projected diffusion estimator}. Thus there are three objects to keep distinct. The exact diffusion estimator $g$ is the geometry-aware statistical target defined by the heat kernel. The projected estimator $g_M$ is a finite spectral approximation to that target. The finite-element estimator $\tilde g_M$ in Section~\ref{sec:pdkde} is the numerical version used in our experiments and application.

The projected representation reveals the key simplification. The eigenpairs $(\lambda_m,u_m)$ are computed once for the domain. Given a new sample, the estimator is updated only through the coefficients $\frac{1}{N}\sum_{i=1}^N u_m(X_i)$, so no new PDE solve is required. Thus the method targets exactly the fixed-domain, changing-data structure of many spatial applications.

Note that although the condition $\int_D g_M(x;t)\,dx = 1$ strictly holds, the finite truncation relaxes the pointwise nonnegativity of the exact diffusion estimator. In computation this is treated as a diagnostic rather than hidden by post-processing: the minimum fitted value and integrated negative part can be monitored as $M$ and the mesh resolution are varied. When these quantities are not negligible, increasing $M$ or refining the mesh is preferable to truncating the surface after estimation, because truncation changes the spectral estimator being analyzed.

\begin{remark}\label{remark: computational complexity of g_M}
    Consider the computational cost of \eqref{eq: diffusion estimator, partial sum} relative to standard KDE. If the kernel of standard KDE can be evaluated in constant time, then the cost of evaluating $\hat f(x)$ at one point is $\Theta(N)$. In contrast, if the coefficients $\frac{1}{N}\sum_{i=1}^N u_m(X_i)$ are precomputed, then evaluating $g_M(x;t)$ at one point costs $\Theta(M)$. The projected estimator also requires only $\Theta(M)$ memory to store these coefficients, whereas standard KDE requires $\Theta(N)$ memory. Thus, as long as $M$ can be chosen substantially smaller than $N$ without sacrificing accuracy, the spectral representation yields a genuine computational gain.
\end{remark}

\subsection{Truncation error and choice of cutoff}\label{sec:truncation}

The purpose of this subsection is to connect the computational approximation to the statistical target. The following results quantify the truncation error incurred by replacing $g$ with $g_M$ and indicate how large $M$ must be for the projected estimator to behave like the full diffusion estimator.

\begin{lemma}[Uniform bound on the truncation error]\label{lemma: g_M, L infty bound}
    Let $C>0$ be a constant that depends only on the domain $D$. When $M^{\frac{2}{d}}t$ is sufficiently large, we have
    \begin{equation}
        \|g(\cdot;t) - g_M(\cdot;t)\|_\infty \lesssim \exp \left( \left(\frac{1}{2} - d\right)\log t - \frac{C}{2} M^{\frac{2}{d}} t \right), \label{eq: bound of g-g_M, Linfty}
    \end{equation}
    with probability one.
\end{lemma}

The proof is given in Appendix~\ref{SM:g_M, L infty bound}.

Because the truncation error is uniformly controlled, $g_M$ inherits the large-sample behavior of the diffusion estimator $g$ when $M$ increases at a suitable rate. The following theorem establishes the global consistency of $g_M$, guaranteeing that its mean integrated squared error (MISE) converges to zero over the entire domain.

\begin{theorem}[Mean integrated squared error of $g_M$]\label{proposition: MISE of g_M}
    Let $f\in C(\overline D)$. Assume that the bandwidth $\sqrt{t_N}$ and cutoff $M_N$ satisfy $t_N\rightarrow 0$, $Nt_N^{d/2} \rightarrow \infty$, and $\frac{M_N^{2/d} t_N}{\log (t_N^{-1})}\rightarrow\infty$ as $N\rightarrow\infty$. Then $\Exp \|g(\cdot;t_N) - g_{M_N}(\cdot;t_N)\|_{L^2(D)}^2\rightarrow 0$ and $\MISE(g_{M_N}(\cdot;t_N)) := \Exp \|g_{M_N}(\cdot;t_N) - f\|_{L^2(D)}^2\rightarrow 0$.
\end{theorem}

The proof is given in Appendix~\ref{proof: MISE of g_M}.

Furthermore, the following theorem guarantees the pointwise convergence of $g_M$ to the true density everywhere, even at the boundary.

\begin{theorem}[Pointwise consistency of the projected diffusion estimator]\label{proposition: g_M, weak consistency}
    Let $f\in C(\overline D)$. Assume that the bandwidth $\sqrt{t_N}$ and cutoff $M_N$ satisfy $t_N\rightarrow 0$, $Nt_N^{d/2} \rightarrow \infty$, and $\frac{M_N^{2/d} t_N}{\log (t_N^{-1})}\rightarrow\infty$ as $N\rightarrow\infty$. Then $\Exp |g_{M_N}(x;t_N) - f(x)|^2$ converges uniformly to $0$ on $\overline D$. Moreover, for any $x\in\overline{D}$,
    \begin{equation*}
        g_{M_N}(x;t_N)-g(x;t_N)\rightarrow_p 0,
    \end{equation*}
    and hence $g_{M_N}(x;t_N)\rightarrow_p f(x)$.
\end{theorem}

The proof is given in Appendix~\ref{proof: g_M, weak consistency}.

A useful implication is obtained under the optimal bandwidth scaling $t_N = \Theta \left( N^{-\frac{2}{d+4}} \right)$. In that case the condition $\frac{M_N^{2/d} t_N}{\log (t_N^{-1})}\rightarrow\infty$ becomes
\begin{equation*}
\frac{M_N^{2/d}}{N^\frac{2}{d+4}\log N}\rightarrow\infty,
\end{equation*}
so it is sufficient to choose
\begin{equation*}
M_N=\omega\left(N^{\frac{d}{d+4}}(\log N)^{\frac{d}{2}}\right).
\end{equation*}
Combined with Remark~\ref{remark: computational complexity of g_M}, this shows that the projected estimator can reduce the cost of density evaluation by an order in the sample size while preserving the asymptotic behavior of the full diffusion estimator. In finite samples, $M$ should be treated as a numerical resolution parameter: it must be large enough that increasing $M$ no longer changes the fitted surface, the selected bandwidth, or boundary diagnostics at the scale of scientific interpretation, but it need not be comparable to $N$. The experiments below use $M=512$; the resulting cost comparison should therefore be read as the cost of using a fixed spectral resolution after the domain has been preprocessed.

\begin{example}
    Consider the one-dimensional domain $D=(0,1)$. The eigenvalues and eigenfunctions of \eqref{eq: eigenvalue, eigenfunction} are
    \begin{equation}
    \begin{gathered}
        \lambda_m = m^2 \pi^2 \quad (m\geq 0),\\
        u_0(x) = 1,\quad u_m(x) = \sqrt{2} \cos (m\pi x) \quad (m\geq 1).
    \end{gathered}\label{eq: eigenfunctions, 1dim}
    \end{equation}
    For convenience the indexing here starts at $m=0$. The projected diffusion estimator becomes
    \begin{equation}
        g_M(x;t) = 1 + \sum_{m=1}^{M-1} e^{-\frac{1}{2}(m\pi)^2t} \sqrt{2}\cos(m\pi x) \frac{1}{N} \sum_{i=1}^N \sqrt{2}\cos(m\pi X_i).
    \end{equation}
    If $t = \Theta(N^{-\frac{2}{5}})$ and $M = \Theta(N^{\frac{1}{5}}\log N)$, then Theorems~\ref{proposition: MISE of g_M} and \ref{proposition: g_M, weak consistency} imply that $g_M$ is MISE-consistent and pointwise consistent up to the boundary, while the computational cost of evaluating $g_M(x;t)$ at a point is only $\Theta(N^{\frac{1}{5}}\log N)$.
\end{example}

\begin{example}
    Let $D\subset \mathbb{R}^d$ and suppose the diffusion bandwidth is chosen at its MISE-optimal order $t_N = \Theta(N^{-\frac{2}{d+4}})$. Then the diffusion estimator attains the usual rate $\MISE(g(\cdot;t_N)) = O(N^{-\frac{4}{d+4}})$ under some norm constraint on $f$; see Theorem~\ref{corollary: MISE of g, Rf} in the appendices. Fix any $\varepsilon>0$ and take $M_N = \Theta(N^{\frac{d}{d+4}+\varepsilon})$. By Lemma~\ref{lemma: g_M, L infty bound},
    \begin{align}
        \Exp\|g(\cdot;t_N) - g_{M_N}(\cdot;t_N)\|_{L^2(D)}^2 &= O \left( t_N^{1-2d} e^{- C M_N^{2/d} t_N} \right) \\
        &= O\left( N^\frac{2(2d-1)}{d+4} e^{-CN^{2 \varepsilon / d}} \right).
    \end{align}
    Therefore,
    \begin{align}
        \MISE(g_{M_N}(\cdot;t_N)) &\lesssim \MISE(g(\cdot;t_N)) + \Exp\|g(\cdot;t_N) - g_{M_N}(\cdot;t_N)\|_{L^2(D)}^2\\
        &= O\left(N^{-\frac{4}{d+4}}\right) + O\left( N^\frac{2(2d-1)}{d+4} e^{-CN^{2 \varepsilon / d}} \right)\\
        &= O\left(N^{-\frac{4}{d+4}}\right).
    \end{align}
    Thus the projected estimator preserves the MISE order of the full diffusion estimator while reducing the evaluation cost to $\Theta(N^{\frac{d}{d+4}+\varepsilon})$.
\end{example}

\subsection{Bandwidth selection}

A practical advantage of PDKDE is that the bandwidth can be selected from an analytic criterion once the spectral coefficients are available. We use least-squares cross-validation (LSCV), which allows the smoothing level to be chosen from the data without rerunning a PDE solver for each candidate bandwidth.

Our goal is to choose $\sqrt{t}$ to minimize the integrated squared error,
\begin{equation*}
\int_D |g_M(x;t)-f(x)|^2 dx = \int_D g_M^2(x;t)dx - 2 \int_D g_M(x;t)f(x)dx + \int_D f^2(x) dx.
\end{equation*}
The first term is explicit, and the second can be approximated by leave-one-out cross-validation. This yields
\begin{align}
    \operatorname{LSCV}(t) &= \frac{1}{N^2} \sum_{m=1}^M e^{-\lambda_m t} \left( \sum_{i=1}^N u_m(X_i) \right)^2 \nonumber \\*
    &\quad - \frac{2}{N(N-1)} \sum_{m=1}^M e^{-\frac{1}{2}\lambda_m t} \left\{ \left(\sum_{i=1}^N u_m(X_i)\right)^2 - \sum_{i=1}^N u_m^2(X_i) \right\}. \label{eq: projected diffusion estimator, cv risk}
\end{align}
A derivation is given in Appendix~\ref{SM: bandwidth selection}.

Minimizing $\operatorname{LSCV}(t)$ is theoretically well founded because
\begin{equation*}
\Exp[\operatorname{LSCV}(t)] = \MISE(g_M(\cdot;t)) - \|f\|_{L^2(D)}^2,
\end{equation*}
so the criterion is an unbiased estimator of the MISE up to an additive constant.

From a computational viewpoint, the same cached quantities that define the estimator also define the cross-validation score. If $\sum_{i=1}^N u_m(X_i)$ and $\sum_{i=1}^N u_m^2(X_i)$ are precomputed for each $m\in \{1,2,\dotsc,M\}$, then evaluating $\operatorname{LSCV}(t)$ at a given bandwidth costs only $O(M)$. By the truncation results in Section~\ref{sec:truncation}, $M$ can be chosen much smaller than $N$ while preserving the asymptotic properties of the estimator. Moreover, because $\operatorname{LSCV}(t)$ is analytic in $t$, it can be optimized by standard one-dimensional numerical methods without repeatedly solving a PDE. In the numerical work we optimize over $\log t$ by a bounded one-dimensional search on a fixed interval scaled to the squared diameter of the domain; the same interval, optimizer, and convergence tolerance are used for all PDKDE-LSCV fits and are reported in the replication code. The reported comparison therefore separates the bandwidth rule from the estimator: PDKDE-ISJ uses the same external plug-in rule as the Euclidean-kernel baselines, while PDKDE-LSCV uses the domain-aware analytic criterion derived above.

\subsection{Finite-element implementation}

If the domain is simple, for example $D=[0,1]^d$, the eigenpairs in \eqref{eq: eigenvalue, eigenfunction} can sometimes be written explicitly. For general domains this is no longer possible, so we compute them numerically by the finite element method \citep{boffi2010finite}.

Let $\tilde V$ be a finite-dimensional function space spanned by basis functions $\{\psi_k\}_{k=1}^K$. Any $u\in \tilde V$ can be written as
\begin{equation}
    u(x) = \sum_{k=1}^K u^k \psi_k(x),
\end{equation}
for coefficients $u^k\in\mathbb{R}$. The finite-element approximation to the eigenproblem is to find $u\in \tilde V$ and $\lambda\in\mathbb{R}$ such that
\begin{equation}
    \int_D \langle \nabla u, \nabla v \rangle dx = \lambda \int_D u v dx \qquad (\forall v\in \tilde V).
\end{equation}
Let $\{(\tilde{\bm{u}}_k, \tilde \lambda_k)\}_{k=1}^K$ denote the resulting eigenpairs. The projected estimator is then approximated by
\begin{equation}
    \tilde g_M(x;t) = \sum_{m=1}^M e^{-\frac{1}{2} \tilde \lambda_m t} \left( \sum_{k=1}^K \tilde{u}_m^k \psi_k(x) \right) \left( \frac{1}{N} \sum_{i=1}^N \sum_{k=1}^K \tilde{u}_m^k \psi_k(X_i) \right), \label{eq: fem projected diffusion estimator}
\end{equation}
where $M \leq K$.

Equation \eqref{eq: fem projected diffusion estimator} is the version used in our experiments and application. In practice, the domain computation is performed once: one constructs a mesh and solves the generalized eigenvalue problem. After that, each new estimate requires only the projected sample moments and evaluation of \eqref{eq: fem projected diffusion estimator} for the desired bandwidth and grid. This one-time domain computation is central to the practical value of PDKDE. All simulations and the Oahu application therefore use the finite-element estimator $\tilde g_M$ in \eqref{eq: fem projected diffusion estimator}. The consistency results in Section~\ref{sec:truncation} are stated for the exact projected estimator $g_M$; the finite-element approximation is a numerical discretization of that estimator.

The additional discretization error incurred by replacing $g_M$ with $\tilde g_M$ can be analyzed using standard finite-element approximation theory for elliptic eigenvalue problems \citep{boffi2010finite, babuska1989finite}. For applied use, the important check is that the mesh resolves the boundary features and eigenfunctions retained in the truncation. We therefore keep the mesh fixed within each comparison, report the mesh parameters used to construct the eigenbasis, and interpret changes in the fitted maps under mesh refinement as numerical sensitivity rather than statistical variability.

More implementation details, including the finite-element settings and cross-validation implementation, are given in the appendices.

\section{Numerical Experiments}\label{sec:experiments}

We now examine whether the projected representation solves the applied failures that motivate the paper. The experiments are designed as controlled stress tests for three properties needed in spatial intensity mapping on complex domains: (i) accurate estimation near boundaries, (ii) resistance to erroneous smoothing across geometric barriers, and (iii) fast repeated computation after the domain has been fixed.

\subsection{Experimental setting}

We consider two-dimensional domains with unit area, $D\subset\mathbb{R}^2$, that capture two geometric stress tests most relevant to the Oahu application. Simulation 1 studies boundary concentration on a convex polygon, where leakage and boundary bias dominate. Simulation 2 considers a non-convex U-shaped domain, where high-density and low-density regions can be close in Euclidean distance while remaining well separated by the domain geometry. Together these settings test near-boundary hotspots and smoothing across an artificial barrier; additional domain geometries such as holes and narrow corridors would be natural extensions, but the present experiments isolate the two failures most directly visible in the application. Figure~\ref{fig: domain and data} shows the two settings.

\begin{figure}[t]
    \centering
    \includegraphics[width=0.7\textwidth]{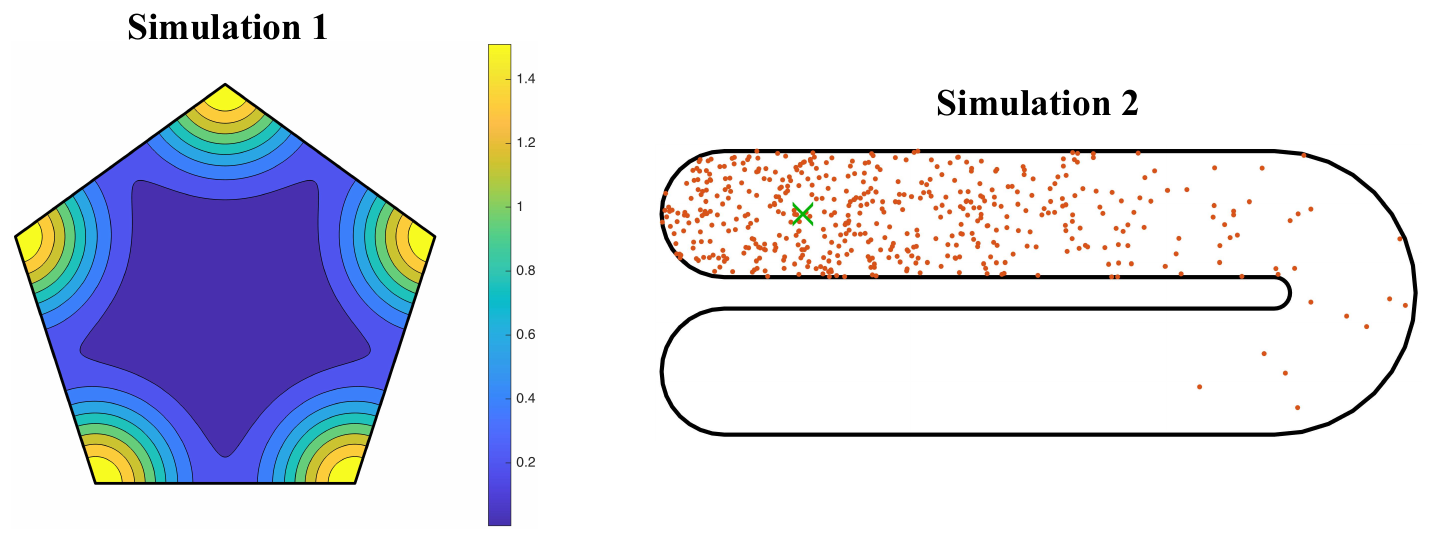}
    \caption{Domains and data-generating mechanisms used in the simulations. Simulation 1 samples from a Gaussian mixture truncated to a pentagon. Simulation 2 starts at the green dot and samples from Brownian motion reflected at the boundary.}
    \label{fig: domain and data}
\end{figure}

We compare the following methods.

\begin{enumerate}
    \item \textbf{Gaussian KDE (GKDE)}: standard Gaussian KDE with bandwidth chosen by the Improved Sheather-Jones (ISJ) method.
    \item \textbf{Renormalized Gaussian KDE (RGKDE)}: Gaussian KDE with the renormalization \eqref{eq: renormalization kernel}. The bandwidth is selected by ISJ, and the normalizing constants are approximated by 10,000 Monte Carlo samples.
    \item \textbf{Density Estimator with Partial Differential Equation (DE-PDE)} \citep{ferraccioli2021nonparametric}: penalized-spline density estimation implemented with the \texttt{fdaPDE} package in R \citep{fdaPDE}. The regularization parameter is selected by 3-fold cross-validation. The candidate set is $\{4/10,\allowbreak 1/10,\allowbreak 1/40,\allowbreak 1/160,\allowbreak 1/640\}$. We use quadratic elements with maximum triangle area $0.001$ and minimum angle $30$ degrees.
    \item \textbf{Projected Diffusion KDE (PDKDE)}: our estimator \eqref{eq: fem projected diffusion estimator}. We fix $M=512$ and consider two bandwidth strategies: ISJ as an external plug-in rule (PDKDE-ISJ) and the proposed least-squares cross-validation criterion (PDKDE-LSCV), optimized over $\log t$ as described in Section~\ref{sec:pdkde}. The finite-element mesh is constructed in the MATLAB PDE toolbox using quadratic elements with maximum edge length $0.01$.
\end{enumerate}

For accuracy, we use the integrated squared error
\begin{equation}
    \mathrm{ISE}(\hat f, f) = \int_D |\hat f (x) - f(x)|^2 dx.
\end{equation}
In settings where $f$ is not available in closed form, we report a Monte Carlo approximation of
\begin{equation*}
\mathrm{ISE}(\hat f, f)-\|f\|_{L^2(D)}^2 = \int_D \hat f(x)^2 dx - 2 \int_D \hat f(x) f(x) dx,
\end{equation*}
which preserves the ranking across estimators because $\|f\|_{L^2(D)}^2$ is constant across methods. To assess behavior specifically at the boundary, we also use
\begin{equation}
    \mathrm{ISE}_{\mathrm{boundary}}(\hat f, f) = \Exp_{x\sim U(\partial D)} |\hat f (x) - f(x)|^2,
\end{equation}
where $U(\partial D)$ denotes the uniform distribution on the boundary with respect to arc length.

To compare computation times, we consider four practically relevant queries. Given a domain $D$, dataset $S$, bandwidth parameter $t$, and evaluation set $P$, let $E = \{\hat f(x)\,|\,x\in P\}$.
\begin{enumerate}
    \item \textbf{Query A}: given $(D,S,t,P)$, compute $E$.
    \item \textbf{Query B}: fix $D$ (all domain-specific preprocessing has been done), and given $(S,t,P)$, compute $E$.
    \item \textbf{Query C}: fix $D$ and $S$, and given $(t,P)$, compute $E$.
    \item \textbf{Query D}: fix $D$, and given $(S,P)$, select the hyperparameter from the data and compute $E$.
\end{enumerate}
Query A measures end-to-end construction with no domain preprocessing. Queries B and C isolate the benefit of reusing the domain basis when the data, bandwidth, or evaluation points change. Query D measures the total time for automated smoothing, which is especially relevant when many maps must be produced for the same study region.

\subsection{Results from synthetic data}

Figure~\ref{fig: estimator visualization} provides a qualitative comparison on the U-shaped domain. GKDE and RGKDE smooth across the cavity because they rely on Euclidean distance, which produces a map that is incompatible with the geometry of the support. DE-PDE and PDKDE both respect the geometry, but PDKDE retains a direct analytic representation of both the estimator and the bandwidth criterion.

\begin{figure}
    \centering
    \includegraphics[width=0.85\textwidth]{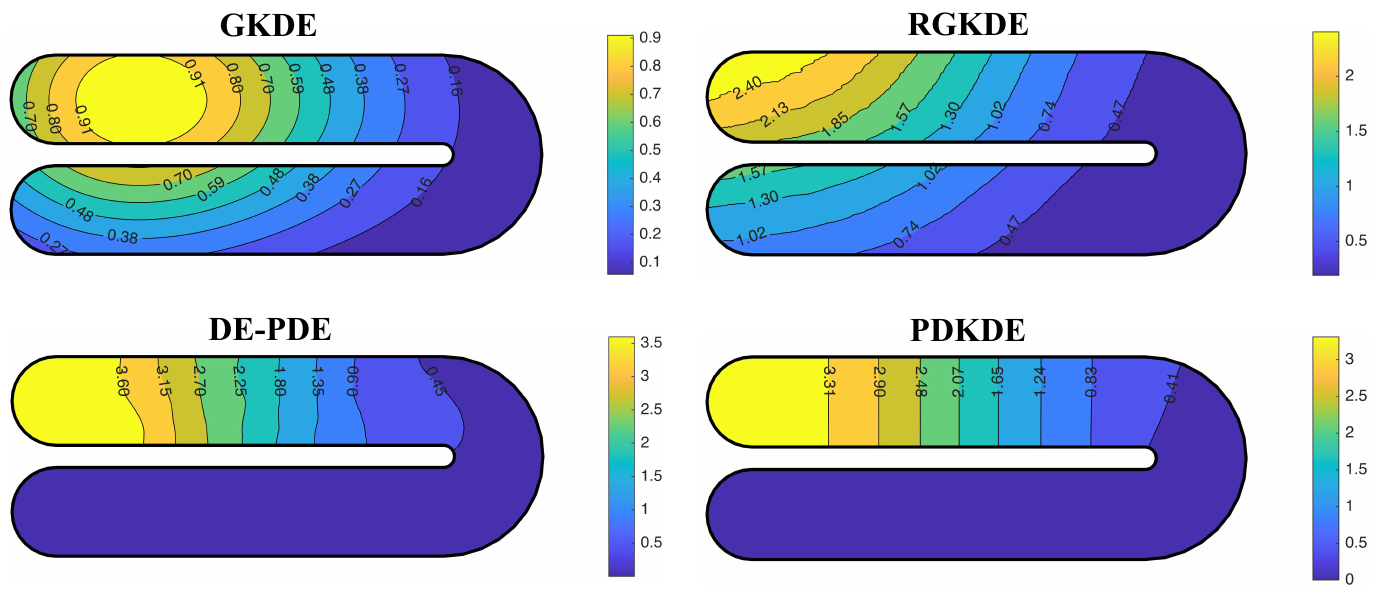}
    \caption{Estimated densities for a dataset of size $N=500$ from Figure~\ref{fig: domain and data}. For GKDE, RGKDE, and PDKDE the bandwidth parameter is fixed at $t=0.1$; for DE-PDE the regularization parameter is $\lambda=0.1$.}
    \label{fig: estimator visualization}
\end{figure}

Tables~\ref{table: ISE, convex}--\ref{table: ISE, ushape1} provide quantitative results. PDKDE-LSCV performs best in all reported settings. The gain over PDKDE-ISJ is notable, especially on the non-convex domain, which suggests that a domain-aware cross-validation criterion is preferable to a Euclidean plug-in bandwidth rule when the support is irregular. Table~\ref{table: boundary ISE, convex} further shows that this advantage is especially pronounced at the boundary. Because PDKDE is a finite spectral estimator, these accuracy summaries should be interpreted together with the fixed truncation level and mesh settings above; in applications, $M$ and mesh refinement should be checked until maps and selected bandwidths are stable.

The computational results in Table~\ref{table: computation time} highlight where the projected formulation matters most. Query A includes the one-time domain computation, so GKDE remains fastest in that fully cold-start scenario. Once the domain basis is available, however, PDKDE becomes the fastest method for Queries B--D and is dramatically faster than DE-PDE and RGKDE. These are exactly the scenarios relevant for repeated analyses on a fixed spatial domain, including the time-binned Oahu application below.

The experiments therefore support the main design goal of the method. PDKDE combines accurate boundary-aware estimation with a substantial reduction in computation after the domain computation, bridging the gap between the statistical appeal of diffusion KDE and the practical demands of applied spatial analysis.

\begin{table}[t]
    \centering
    \caption{Simulation 1: mean (standard deviation) of $\operatorname{ISE}(\hat f, f)$ over 20 datasets. Best results are in bold.}
    \label{table: ISE, convex}
    \large{\begin{tabular}{lccc}
        \hline
        Method & $N=100$ & $N=1000$ & $N=10000$ \\ \hline
        GKDE & 0.4534 (0.0747) & 0.1824 (0.0261) & 0.0781 (0.0044) \\
        RGKDE & 0.2176 (0.0602) & 0.0750 (0.0168) & 0.0229 (0.0026) \\
        DE-PDE & 0.2627 (0.0665) & 0.0635 (0.0183) & 0.0138 (0.0029) \\
        PDKDE-ISJ & 0.2191 (0.0701) & 0.0818 (0.0181) & 0.0256 (0.0028) \\
        PDKDE-LSCV & \textbf{0.2102} (0.0633) & \textbf{0.0542} (0.0158) & \textbf{0.0122} (0.0020) \\
        \hline
    \end{tabular}}
\end{table}

\begin{table}[t]
    \centering
    \caption{Simulation 1: mean (standard deviation) of $\operatorname{ISE}_{\mathrm{boundary}}(\hat f, f)$ over 20 datasets. Best results are in bold.}
    \label{table: boundary ISE, convex}
    \large{\begin{tabular}{lccc}
        \hline
        Method & $N=100$ & $N=1000$ & $N=10000$ \\ \hline
        GKDE & 2.6348 (0.2633) & 2.0995 (0.1822) & 1.9787 (0.0499) \\
        RGKDE & 0.7506 (0.2774) & 0.3608 (0.1413) & 0.1216 (0.0441) \\
        DE-PDE & 0.8331 (0.2429) & 0.2131 (0.0835) & 0.0639 (0.0283) \\
        PDKDE-ISJ & 0.7060 (0.2778) & 0.3406 (0.1350) & 0.0940 (0.0270) \\
        PDKDE-LSCV & \textbf{0.6805} (0.2335) & \textbf{0.2047} (0.0808) & \textbf{0.0418} (0.0105) \\
        \hline
    \end{tabular}}
\end{table}

\begin{table}[t]
    \centering
    \caption{Simulation 2: mean (standard deviation) of the Monte Carlo approximation to $\operatorname{ISE}(\hat f, f) - \|f\|_{L^2(D)}^2$ over 20 datasets. Best results are in bold.}
    \label{table: ISE, ushape1}
    \large{\begin{tabular}{lccc}
        \hline
        Method & $N=100$ & $N=1000$ & $N=10000$ \\ \hline
        GKDE & -2.5308 (0.0582) & -2.7438 (0.0168) & -2.8228 (0.0065) \\
        RGKDE & -2.6905 (0.0673) & -2.8293 (0.0196) & -2.8726 (0.0039) \\
        DE-PDE & -2.8053 (0.0406) & -2.8895 (0.0046) & -2.8940 (0.0022) \\
        PDKDE-ISJ & -2.7451 (0.0596) & -2.8502 (0.0177) & -2.8865 (0.0032) \\
        PDKDE-LSCV & \textbf{-2.8337} (0.0406) & \textbf{-2.8916} (0.0058) & \textbf{-2.8958} (0.0011) \\
        \hline
    \end{tabular}}
\end{table}

\begin{table}[t]
    \centering
    \caption{Computation time for each query (seconds), with $N=20000$ and $|P|=20000$. For Query D, GKDE and RGKDE use ISJ bandwidth selection, PDKDE uses LSCV, and DE-PDE selects from five candidates using 3-fold cross-validation. Fastest results are in bold.}
    \label{table: computation time}
    \large{\begin{tabular}{lrrrr}
        Method & Query A & Query B & Query C & Query D \\ \hline
        GKDE & \textbf{7.48} & 7.48 & 7.48 & 8.05 \\
        RGKDE & 154.59 & 154.59 & 154.59 & 256.93 \\
        DE-PDE & 31.65 & 31.65 & 31.65 & 232.79 \\
        PDKDE & 27.62 & \textbf{0.80} & \textbf{0.26} & \textbf{0.83} \\\hline
    \end{tabular}}
\end{table}

\section{Application: Oahu theft and larceny intensity mapping}\label{sec:application}

We now return to the motivating application: estimating the spatial intensity of theft and larceny incidents on Oahu, Hawaii, from 11/27/2024 to 5/28/2025. The data were obtained from the Honolulu Police Department\footnote{\url{https://data.honolulu.gov/Public-Safety/HPD-Crime-Incidents/vg88-5rn5/about_data}; accessed May 29, 2026.}. Our goal is descriptive spatial density estimation: we seek maps that faithfully summarize where reported incidents are concentrated on the island. We do not use the analysis to explain the causes of incidents or to prescribe enforcement decisions.

This application is demanding for a density estimator for two reasons. First, the coastline is a hard support constraint, and many incidents occur close to it. A method that pulls density inland or assigns mass to the ocean changes the interpretation of coastal hotspots. Second, the domain is fixed but the relevant data subset changes when the analysis is stratified by time of day. A useful method must therefore respect the island geometry and allow repeated estimation without repeated PDE solves.

For reproducibility, the analysis uses HPD Crime Incidents records from the City and County of Honolulu open data portal with \texttt{type} equal to \texttt{THEFT/LARCENY} and event dates from 2024-11-27 through 2025-05-28. Records with missing or invalid event times, longitude, or latitude are removed, and repeated coordinates are retained as repeated reported incidents rather than jittered. Times are grouped into three-hour intervals in local Hawaii time. Longitude--latitude coordinates and the Oahu support polygon are projected to WGS 84 / UTM zone 4N (EPSG:32604) before mesh construction and density evaluation; the same projected coordinates are used for all methods. The raw download, access date, support polygon, point-filtering script, mesh file, and $200\times 200$ evaluation grid are fixed across all time windows and are included in the replication archive. The Oahu domain is constructed once from this projected coastline polygon, with no data-dependent changes to the boundary. The Oahu mesh is constructed in the MATLAB PDE toolbox using the same element order and mesh-resolution rule as in Section~\ref{sec:experiments}, the first $M=512$ Neumann eigenpairs are retained, and the same basis is reused for every time slot.

We divide the observations into eight three-hour time slots, $S_1,\dotsc,S_8$, with $n_j = |S_j|$, and estimate a separate conditional spatial density $\hat f_j(x)$ for each slot. To compare activity across time windows, we visualize spatial intensity surfaces proportional to $n_j\hat f_j(x)$, so that both the within-slot spatial distribution and the slot-specific event count are reflected. For visualization we evaluate the estimators on a $200\times 200$ grid and retain grid points inside the Oahu support. The computational settings are the same as in the simulation study, except that the candidate regularization parameters for DE-PDE are rescaled to $\{0.005\cdot 4, 0.005, 0.005/4, 0.005/16, 0.005/64\}$ to reduce severe oversmoothing. All randomized components used in the numerical approximations, such as Monte Carlo normalizers and cross-validation folds, are run with fixed seeds documented in the replication code; software versions and timing hardware are also reported there.

Figure~\ref{figure: honolulu, data} shows a clear temporal and spatial pattern. Reported incident intensity rises during the day, peaks around 12:00--15:00, and decreases after midnight. Spatially, incidents are concentrated around Waikiki Beach and a few other coastal areas. The coastline is therefore not a minor detail of the problem; it is central to the interpretation of the resulting hotspot maps.

The computational advantage of PDKDE is visible in the time-window analysis. Table~\ref{table: oahu runtime} separates the one-time domain computation from the repeated estimation cost. Across all eight time slots, total online computation times were 2.2 seconds for GKDE, 38.3 minutes for RGKDE, 33.7 minutes for DE-PDE, and 2.9 seconds for PDKDE, after a separate one-time PDKDE domain computation of 31.6 seconds. Thus, once the domain basis has been computed, PDKDE updates the estimate in about 0.4 seconds per time frame. Even including the one-time eigenbasis computation, the eight-map PDKDE analysis takes less than one minute, while the geometry-aware competitors take more than half an hour. This timing is important because boundary-aware alternatives are useful only if they can be recomputed at the scale at which the spatial question is asked.

\begin{table}[tbp]
    \centering
    \caption{Oahu timing decomposition for eight three-hour maps. For PDKDE, the one-time domain cost includes construction of the Oahu mesh and computation of the first $M=512$ Neumann eigenpairs. Online time includes coefficient updates, bandwidth selection, and grid evaluation for the eight time windows.}
    \label{table: oahu runtime}
    \begin{tabular}{lrrr}
    \hline
    Method & One-time domain cost & Online time for 8 maps & Total time used for maps \\ \hline
    GKDE & -- & 2.2 s & 2.2 s \\
    RGKDE & -- & 38.3 min & 38.3 min \\
    DE-PDE & -- & 33.7 min & 33.7 min \\
    PDKDE & 31.6 s & 2.9 s & 34.5 s \\ \hline
    \end{tabular}
\end{table}

Figure~\ref{figure: honolulu, pdkde} shows the PDKDE-LSCV estimates over the full day. The estimator places intensity along the coastline when supported by the data while still producing smooth and interpretable hotspot maps. Results for the competing methods are provided in Figures~\ref{figure: honolulu, gkde}--\ref{figure: honolulu, depde} in the appendices.

To compare the methods visually, Figures~\ref{figure: honolulu, waikiki} and \ref{figure: honolulu, 18} focus on the 18:00--21:00 time window. The Waikiki zoom-in in Figure~\ref{figure: honolulu, waikiki} gives a clear example of boundary bias in GKDE: density is pulled away from the coast, creating an artificial low-density strip near the boundary. The boundary-aware methods place mass directly along the coastline, but they differ in smoothness and geometric fidelity.

The whole-island comparison in Figure~\ref{figure: honolulu, 18} highlights the practical effect of boundary handling and bandwidth selection. Under the ISJ bandwidth used here, GKDE and RGKDE produce rougher surfaces with several isolated local peaks, consistent with the fact that the bandwidth rule is not adapted to the constrained domain. DE-PDE is smoother, but with the penalty grid used here it suppresses some smaller clusters away from the main southern cluster. PDKDE-LSCV gives a coastline-constrained descriptive surface that remains smooth while retaining localized clusters visible in the data.

Overall, the Oahu application shows PDKDE in the setting for which it is intended. It produces coastline-aware descriptive intensity maps, reduces the most visible artifacts of Euclidean smoothing in this application, and makes repeated boundary-aware estimation feasible on a realistic geographic domain. Code for the simulation study, Oahu preprocessing, mesh construction, eigenbasis computation, LSCV evaluation, and figure generation will be made available in an anonymized repository for review and in a permanent public repository upon acceptance.

\begin{figure}[tbp]
    \centering
    \includegraphics[width=0.85\textwidth]{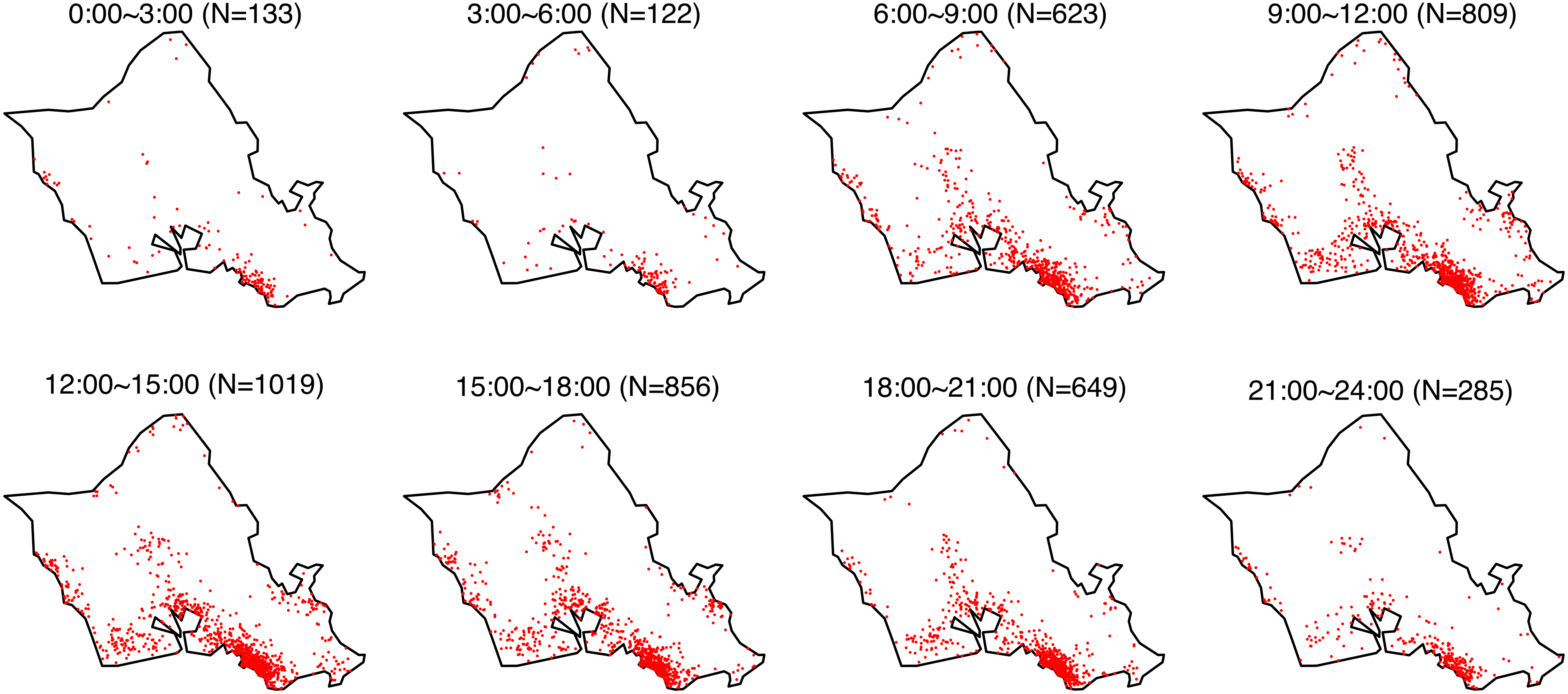}
    \caption{Theft and larceny incidents on Oahu, Hawaii, from 11/27/2024 to 5/28/2025, grouped into three-hour intervals and plotted by location.}
    \label{figure: honolulu, data}
\end{figure}

\begin{figure}[tbp]
    \centering
    \includegraphics[width=0.90\textwidth]{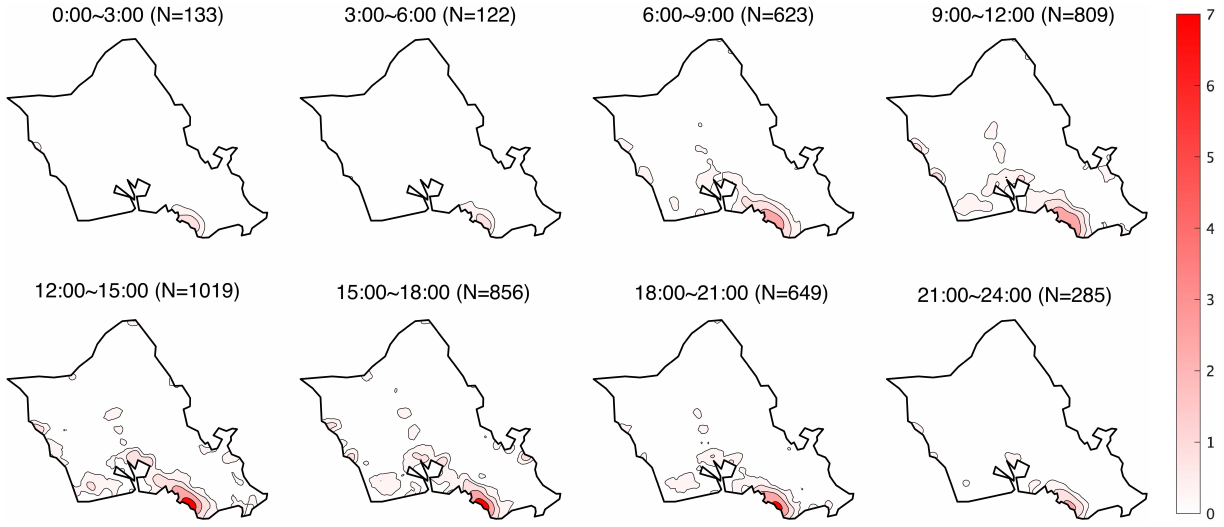}
    \caption{PDKDE estimates for theft and larceny incidents on Oahu, Hawaii, from 11/27/2024 to 5/28/2025, shown separately for each three-hour interval. The maps are constrained to the island domain and preserve intensity along coastal regions.}
    \label{figure: honolulu, pdkde}
\end{figure}

\begin{figure}[tbp]
    \centering
    \includegraphics[width=0.85\textwidth]{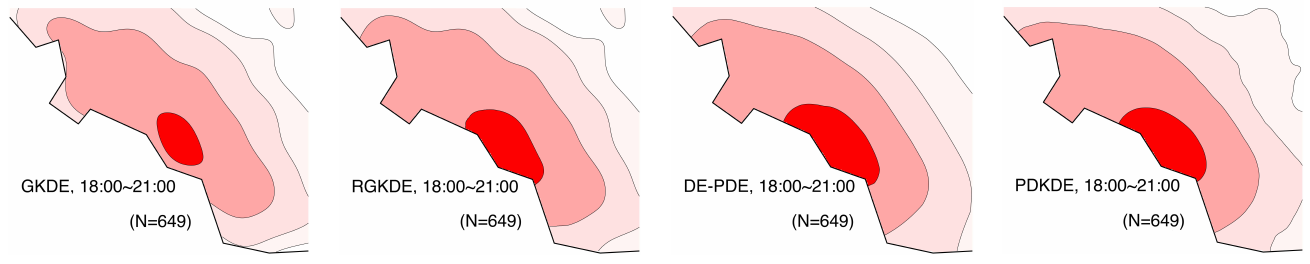}
    \caption{Estimator behavior near Waikiki Beach, 18:00--21:00. From left to right: GKDE, RGKDE, DE-PDE, and PDKDE-LSCV.}
    \label{figure: honolulu, waikiki}
\end{figure}

\begin{figure}[tbp]
    \centering
    \includegraphics[width=0.95\textwidth]{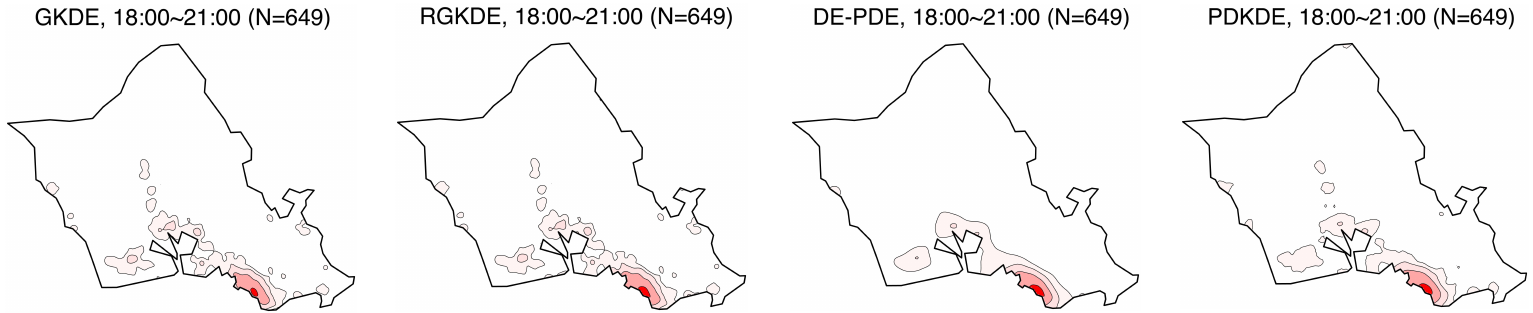}
    \caption{Estimator behavior over Oahu, 18:00--21:00. From left to right: GKDE, RGKDE, DE-PDE, and PDKDE-LSCV.}
    \label{figure: honolulu, 18}
\end{figure}

\section{Discussion}\label{sec:discussion}

We have proposed PDKDE as a boundary-aware method for spatial density estimation on complex domains. The motivating problem is applied: when events are observed on a geographically constrained region, the estimated intensity surface should respect the support, adapt to the domain geometry, and remain fast enough to recompute when the data subset or smoothing level changes. Standard Euclidean kernels fail the first two requirements, while direct diffusion KDE can fail the computational requirement.

PDKDE addresses this gap by keeping the boundary-aware and geometry-respecting features of diffusion KDE while replacing repeated PDE computation with a spectral representation. Once the domain basis is available, density estimation and bandwidth selection reduce to explicit calculations involving projected sample coefficients. In the two simulation settings considered, this construction reduces boundary error and avoids the most visible Euclidean smoothing artifact across the non-convex barrier. The Oahu application shows the same advantage in a real coastal domain, where reported incident intensity is concentrated near the boundary and the maps must be recomputed across time windows.

The method is most useful in settings where the domain is complex but reused: islands, watersheds, lakes, ecological reserves, administrative regions, and urban areas with physical barriers. Its main limitation is the initial eigenbasis computation. For a single small analysis on a simple domain, that cost may not be worthwhile relative to simpler boundary corrections. The method also assumes that the support is fixed; if the domain itself changes over time, the eigenbasis must be recomputed. Very fine coastline detail, many holes, or thin corridors require mesh refinement and therefore larger eigenproblems. Finally, the present construction uses isotropic diffusion on a two-dimensional region. Network-constrained events, anisotropic movement, heterogeneous media, spatiotemporal barriers, and higher-dimensional domains require extensions of the operator and numerical implementation. These limitations do not undermine the Oahu analysis, where the domain is fixed and reused across time windows, but they define the settings in which additional methodological work is needed.

\begin{acks}[Acknowledgements]
Takumi Nakagawa's second affiliation is Fujitsu Limited. Takafumi Kanamori's second affiliation is Center for Advanced Intelligence Project, RIKEN.
This work was partially supported by JSPS KAKENHI Grant Numbers 20H00576, 23H03460, and 24K14849.
\end{acks}

% \clearpage

\appendix

\numberwithin{equation}{section}
\setcounter{equation}{0}
\counterwithin{theorem}{section}
\counterwithin{figure}{section}
\counterwithin{table}{section}

\section{Notation List}
\begin{itemize}
    \item $\Exp [X]$: expectation of $X$
    \item $\Var [X]$: variance of $X$
    \item $\Delta_x$: Laplacian with respect to $x$
    \item $\delta(x)$: Dirac delta function
    \item $C(D)$: the set of continuous functions on $D$
    \item $C^2(D)$: the set of twice continuously differentiable functions on $D$
    \item $L^2(D)$: the set of functions such that $\int_D f^2(x)dx < \infty$
    \item $\langle f, g \rangle_{L^2(D)} := \int_D f(x) g(x) dx$
    \item $\|f\|_{L^2(D)}^2 := \langle f, f\rangle_{L^2(D)}$
    \item $L^\infty(D)$: the set of functions such that $\sup_{x\in D} |f(x)| < \infty$
    \item $\|f\|_{\infty} := \sup_{x\in D} |f(x)|$
    \item $a(n) = O(b(n))$ as $n\rightarrow\infty$ means that $\exists C > 0, \exists n_0, \forall n > n_0, a(n) \leq C b(n)$.
    \item $a(n) = o(b(n))$ as $n\rightarrow \infty$ means that $\forall C>0, \exists n_0, \forall n > n_0, a(n) \leq C b(n)$.
    \item $a(n) = \omega(b(n))$ as $n\rightarrow \infty$ means that $\forall C>0, \exists n_0, \forall n > n_0, a(n) \geq C b(n)$.
    \item $a(n) = \Theta(b(n))$ means that $a(n) = O(b(n))$ and $b(n) = O(a(n))$
    \item $\MISE(\hat f):=\Exp\|\hat f-f\|_{L^{2}(D)}^{2}$ for a density estimator $\hat{f}$ of $f$
\end{itemize}

%\newpage

\section{Additional Figures}

Here, we present additional figures related to the real data application.
\begin{figure}[ht]
    \centering
    \includegraphics[width=0.95\textwidth]{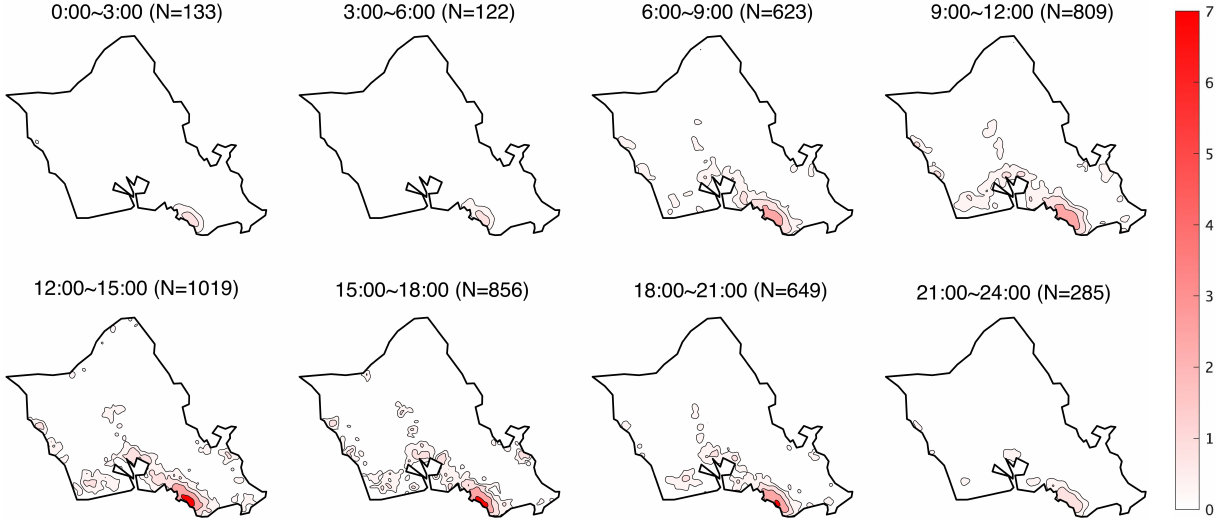}
    \caption{Nonparametric density estimation using GKDE for theft/larceny data from 11/27/2024  to 5/28/2025 in Oahu, Hawaii. Crime occurrences are divided into three hour increments and plotted by location.}
    \label{figure: honolulu, gkde}
\end{figure}

\begin{figure}[ht]
    \centering
    \includegraphics[width=0.95\textwidth]{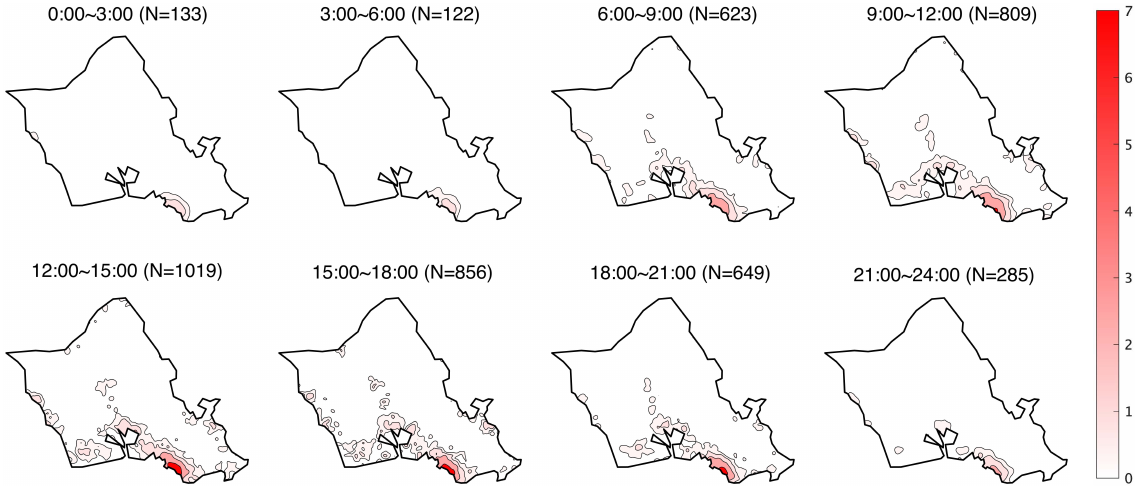}
    \caption{Nonparametric density estimation using RGKDE for theft/larceny data from 11/27/2024  to 5/28/2025 in Oahu, Hawaii. Crime occurrences are divided into three hour increments and plotted by location.}
    \label{figure: honolulu, rgkde}
\end{figure}

\begin{figure}[ht]
    \centering
    \includegraphics[width=0.95\textwidth]{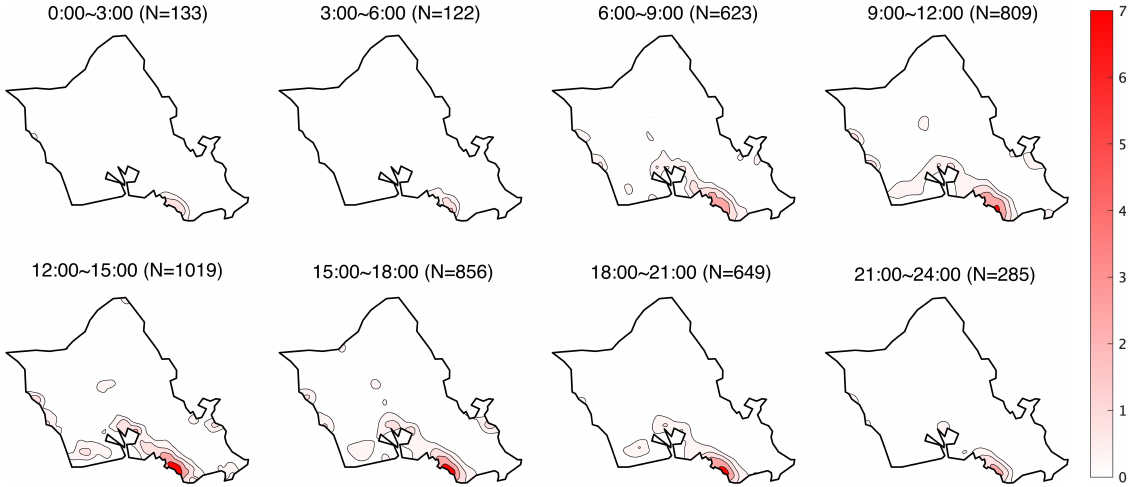}
    \caption{Nonparametric density estimation using DE-PDE for theft/larceny data from 11/27/2024  to 5/28/2025 in Oahu, Hawaii. Crime occurrences are divided into three hour increments and plotted by location.}
    \label{figure: honolulu, depde}
\end{figure}

%\newpage

\section{Bandwidth Selection}\label{SM: bandwidth selection}

Our goal is to select a bandwidth, $\sqrt{t}$, that minimizes the integrated squared error of the estimator, $\int_D |g_M(x;t)-f(x)|^2 dx = \int_D g_M^2(x;t)dx - 2 \int_D g_M(x;t)f(x)dx + \int_D f^2(x) dx$.
The first term can be written as
\begin{align}
    \int_D g_M^2(x;t)dx = \sum_{m=1}^M e^{-\lambda_m t} \left( \frac{1}{N} \sum_{i=1}^N u_m(X_i) \right)^2, \label{eq: lscv_1}
\end{align}
and the second term can be approximated as
\begin{align}
    \int_D g_M(x;t)f(x)dx &\approx \frac{1}{N} \sum_{j=1}^N g_M^{(-j)} (X_j)\\
    &= \frac{1}{N} \sum_{j=1}^N \sum_{m=1}^M e^{-\frac{1}{2}\lambda_m t} u_m(X_j) \frac{1}{N-1} \left( \sum_{i=1}^N u_m(X_i) - u_m(X_j) \right)\\
    &= \frac{1}{N(N-1)} \sum_{m=1}^M e^{-\frac{1}{2}\lambda_m t} \left\{ \left(\sum_{i=1}^N u_m(X_i)\right)^2 - \sum_{i=1}^N u_m^2(X_i) \right\}, \label{eq: lscv_2}
\end{align}
using leave-one-out cross-validation.
Here, the leave-one-out estimator,
\begin{equation}
    g_M^{(-j)} = \sum_{m=1}^M e^{-\frac{1}{2}\lambda_m t} \left(\frac{1}{N-1} \sum_{i\in [N] \setminus \{j\}} u_m(X_i) \right) u_m,
\end{equation}
is an estimator constructed by removing a sample, $X_j$.
From eqs.~\eqref{eq: lscv_1} and \eqref{eq: lscv_2},
\begin{align}
    \operatorname{LSCV}(t) &= \frac{1}{N^2} \sum_{m=1}^M e^{-\lambda_m t} \left( \sum_{i=1}^N u_m(X_i) \right)^2 \nonumber \\ 
    &- \frac{2}{N(N-1)} \sum_{m=1}^M e^{-\frac{1}{2}\lambda_m t} \left\{ \left(\sum_{i=1}^N u_m(X_i)\right)^2 - \sum_{i=1}^N u_m^2(X_i) \right\}
\end{align}
and by solving its minimization,
we can expect to obtain a bandwidth that minimizes the integrated squared error.
We can see that 
    $\operatorname{LSCV}(t)$ is an unbiased estimator of $\MISE (g_M(\cdot;t)) - \|f\|_{L^2(D)}^2$;
    i.e., $\Exp\, [\operatorname{LSCV}(t)] = \MISE (g_M(\cdot;t)) - \|f\|_{L^2(D)}^2$.

\section{Constructing the projected diffusion estimator using the finite element method}

Multiplying the eigenvalue equation $-\Delta u = \lambda u$ by any function $v\in C^\infty (\overline{D})$, and integrating it, we obtain the weak form, 
\begin{equation}
    \int_D \langle \nabla u, \nabla v \rangle dx = \lambda \int_D uv dx \qquad (\forall v), \label{eq: eigenvalue problem, weak form}
\end{equation}
from Green's theorem and the Neumann boundary condition regarding $u$.
Solving  \eqref{eq: eigenvalue problem, weak form} within an infinite-dimensional functional space is difficult, so we constrain it on a finite-dimensional functional space and approximate the eigenvalue and eigenfunction (i.e., the Galerkin method).

Let $\tilde V$ be a functional space spanned by a finite number of functions, $\{\psi_k\}_{k=1}^K$.
Thus, the element of $\tilde V$, $u$, is a function written as
\begin{equation}
    u (x) = \sum_{k=1}^K u^k \psi_k (x)
\end{equation}
using $K$ number of constants, $u^k\in\R$.
Consider constraining $u,v$ in the weak form (\eqref{eq: eigenvalue problem, weak form}) to elements in $\tilde V$.
More specifically, we consider the problem of finding $u = \sum_{k=1}^K u^k \psi_k \in \tilde V$ and $\lambda \in \R$ that satisfies
\begin{equation}
    \int_D \langle \nabla u, \nabla v \rangle dx = \lambda \int_D uv dx \qquad (\forall v\in\tilde V).
\end{equation}
From the linearity of the equation, we only need to consider each basis regarding $v$, so we solve for 
\begin{equation}
    \int_D \left\langle \nabla \sum_{k=1}^K u^k \psi_{k}, \nabla \psi_{k'} \right\rangle dx = \lambda \int_D \sum_{k=1}^K u^k \psi_k \psi_{k'} dx \qquad (\forall k'=1,\dotsc,K)
\end{equation}
regarding $\bm u = (u^1,\dotsc,u^K)^\top$ and $\lambda$.
Then we have, 
\begin{equation}
    \sum_{k=1}^K u^k \int_D \left\langle \nabla \psi_k, \nabla \psi_{k'} \right\rangle dx = \sum_{k=1}^K u^k \lambda \int_D \psi_k \psi_{k'} dx \qquad (\forall k'=1,\dotsc,K)
\end{equation}
so, for each $k,l=1,\dotsc,K$, using the following matrices,
\begin{align}
    A_{kl} &= \int_D \langle \nabla \psi_k, \nabla \psi_l \rangle dx,\\
    M_{kl} &= \int_D \psi_k \psi_l dx,
\end{align}
we can rewrite the eigenvalue problem as
\begin{equation}
    \bm A \bm u = \lambda \bm M \bm u. \label{eq: fem eigenvalue problem}
\end{equation}
Therefore, we can approximate the infinite-dimensional eigenvalue and eigenfunction problem to the finite-dimensional eigenvalue and eigenvector problem.

Let the solution to  \eqref{eq: fem eigenvalue problem} be $\{(\tilde{\bm{u}}_k, \tilde \lambda_k)\}_{k=1}^K$, we approximate the projected diffusion estimator by
\begin{equation}
    \tilde g_M(x;t) = \sum_{m=1}^M e^{-\frac{1}{2} \tilde \lambda_m t} \left( \sum_{k=1}^K \tilde{u}_m^k \psi_k(x) \right) \left( \frac{1}{N} \sum_{i=1}^N \sum_{k=1}^K \tilde{u}_m^k \psi_k(X_i) \right)
\end{equation}
where $M \leq K$.

In the finite element method, the basis, $\{\psi_k\}_{k=1}^K$ is determined by partitioning the domain, $D$, as a mesh.
For example, for the one-dimensional finite element, we set a finite-dimensional functional space, $\tilde V$, such that $u(x)$ is a linear function within each mesh.

\section{Error Analysis of the Diffusion Kernel Density Estimator}

\subsection{Mean Integrated Squared Error}

We calculate the asymptotic behavior of the MISE of the diffusion
estimator, $g(x;t)$, 
\begin{equation}
\MISE(g(\cdot,t))=\Exp\|g(\cdot;t)-f\|_{L^{2}(D)}^{2}.
\end{equation}
First, using the bias and variance of the MISE of each point, we have
\begin{align*}
b(x) & =\Exp g(x;t)-f(x),\\
\sigma^{2}(x) & =\Exp|g(x;t)-\Exp g(x;t)|^{2}.
\end{align*}
We can decompose the MISE as 
\begin{align}
\MISE(g(\cdot,t)) & =\Exp\int_{D}(g(x;t)-f(x))^{2}dx\nonumber \\
 & =\int_{D}\Exp\left|g(x;t)-\Exp g(x;t)+\Exp g(x;t)-f(x)\right|^{2}dx\nonumber \\
 & =\int_{D}\Exp|g(x;t)-\Exp g(x;t)|^{2}dx+\int_{D}(\Exp g(x;t)-f(x))^{2}dx\nonumber \\
 & =\int_{D}\sigma^{2}(x)dx+\int_{D}b^{2}(x)dx,\label{eq: mise of g, bias, variance}
\end{align}
then evaluate the first and second terms.

Define the value, $R(f)$, determined by $f\in L^{2}(D)$ as 
\begin{equation}
    R(f)=\sum_{m=1}^{\infty}\lambda_{m}^{2}\langle f,u_{m}\rangle_{L^{2}(D)}^{2}. \label{eq: definition of Rf}
\end{equation}
When $f$ satisfies the Neumann boundary condition, we have 
\begin{align*}
R(f) & =\sum_{m=1}^{\infty}\langle f,\Delta u_{m}\rangle_{L^{2}(D)}^{2}\\
 & =\sum_{m=1}^{\infty}\langle\Delta f,u_{m}\rangle_{L^{2}(D)}^{2}\\
 & =\|\Delta f\|_{L^{2}(D)}^{2},
\end{align*}
which equates the Laplacian $L^{2}$ norm. Note that $R(f)$
can be defined for $f$ such that it does not satisfy the Neumann boundary condition.

We evaluate the bias and variance using the following lemma. 

\begin{lemma}[asymptotic behavior of the integrated
squared bias and variance.] $D$ is bounded and let $f\in C(\overline{D})$,
and $b(x)$, $\sigma^{2}\left(x\right)$ be the point-wise bias and
variance of the diffusion estimator, $g(x;t)$. Here, if you choose
the bandwidth, $\sqrt{t_{N}}$, such that $t_{N}\rightarrow0\ (N\rightarrow\infty)$,
as $N\rightarrow\infty$, we have 
\begin{equation}
\int_{D}b^{2}(x)dx=o(1).
\end{equation}
In particular, when $R(f)<\infty$, as $N\rightarrow\infty$, we have
\begin{equation}
\int_{D}b^{2}(x)dx=\frac{t_{N}^{2}}{4}R(f)+o(t_{N}^{2}).
\end{equation}
The integrated squared variance is
\[
\int_{D}\sigma^{2}(x)dx=\frac{1}{N(2\sqrt{\pi t_{N}})^{d}}+o\left(\frac{1}{Nt_{N}^{d/2}}\right).
\]

\end{lemma}

\begin{proof}

The integrated bias is 
\begin{align*}
\int_{D}b^{2}(x)dx & =\int_{D}\left(\Exp g(x;t)-f(x)\right)^{2}dx\\
 & =\int_{D}\left(\E_{Y}\kappa(x;Y;t)-f(x)\right)^{2}dx\\
 & =\int_{D}\left(\int_{D}f(y)\kappa(x;y;t)dy-f(x)\right)^{2}dx\\
 & =\int_{D}\left(f_{t}(x)-f(x)\right)^{2}dx\\
 & =\|f_{t}-f\|_{L^{2}(D)}^{2}.%
\end{align*}
Therefore, as $t\rightarrow0$, $\int_{D}b^{2}(x)dx\rightarrow0$
\citep{ito1957fundamental}.

Further, assume $R(f)<\infty$. The integrated bias is 
\[
\|f_{t}-f\|_{L^{2}(D)}^{2}=t^{2}\left\Vert \frac{f_{t}-f}{t}\right\Vert _{L^{2}(D)}^{2},
\]
but if we consider 
\begin{align*}
\lim_{t\rightarrow0}\left\Vert \frac{f_{t}-f}{t}\right\Vert _{L^{2}(D)}^{2} & =\sum_{m=1}^{\infty}\lim_{t\rightarrow0}\left\{ \frac{e^{-\frac{1}{2}\lambda_{m}t}-1}{t}\,\langle f,u_{m}\rangle_{L^{2}(D)}\right\} ^{2}\\
 & =\sum_{m=1}^{\infty}\left\{ -\frac{1}{2}\lambda_{m}\,\langle f,u_{m}\rangle_{L^{2}(D)}\right\} ^{2}\\
 & =\frac{1}{4}R(f),
\end{align*}
as $t\rightarrow0$, we have 
\[
\int_{D}b^{2}(x)dx=\frac{t^{2}}{4}R(f)+o(t^{2}).
\]

The variance for each point, $x\in D$, using $Y\sim f(x)$ (independent
of data) is 
\begin{align*}
\sigma^{2}(x) & =\Exp\left|\frac{1}{N}\sum_{i=1}^{N}\kappa(x;X_{i};t)-\E_{Y}\kappa(x;Y;t)\right|^{2}\\
 & =\frac{1}{N}\Exp\left|\kappa(x;X_{1};t)-\E_{Y}\kappa(x;Y;t)\right|^{2}\\
 & =\frac{1}{N}\E_{Y}\left|\kappa(x;Y;t)\right|^{2}-\frac{1}{N}\left\{ \E_{Y}\kappa(x;Y;t)\right\} ^{2}\\
 & =\frac{1}{N}\int_{D}f(y)\kappa^{2}(x;y;t)dy-\frac{1}{N}\left\{ \int_{D}f(y)\kappa(x;y;t)dy\right\} ^{2}.
\end{align*}
Therefore, if we define the progression of time as 
\[
f_{t}(x)=\int_{D}f(y)\kappa(x;y;t)dy,
\]
we have 
\begin{align}
\int_{D}\sigma^{2}(x)dx & =\frac{1}{N}\int_{D}f(y)\int_{D}\kappa^{2}(x;y;t)dxdy-\frac{1}{N}\int_{D}\left\{ \int_{D}f(y)\kappa(x;y;t)dy\right\} ^{2}dx\nonumber \\
 & =\frac{1}{N}\int_{D}f(y)\kappa(y;y;2t)dy-\frac{1}{N}\|f_{t}\|_{L^{2}(D)}^{2}.\label{eq: variance of g}
\end{align}
The first term is the asymptotic expansion of the heat kernel \citep{branson1990asymptotics},
and, as $t\rightarrow0$, is 
\[
\int_{D}f(y)\kappa(y;y;2t)dy=\frac{1}{(2\sqrt{\pi t})^{d}}+o(t^{-d/2}),
\]
and the second term is 
\begin{align*}
\|f_{t}\|_{L^{2}(D)}^{2} & =\sum_{m=1}^{\infty}e^{-\lambda_{m}t}\langle f,u_{m}\rangle_{L^{2}(D)}^{2}\\
 & <\sum_{m=1}^{\infty}\langle f,u_{m}\rangle_{L^{2}(D)}^{2}\\
 & =\|f\|_{L^{2}(D)}^{2}.
\end{align*}
When $t_{N}\rightarrow0$, $Nt_{N}^{d/2}\rightarrow\infty\ (N\rightarrow\infty)$
is satisfied, as $N\rightarrow\infty$, we have 
\[
\int_{D}\sigma^{2}(x)dx=\frac{1}{N(2\sqrt{\pi t_{N}})^{d}}+o\left(\frac{1}{Nt_{N}^{d/2}}\right).
\]

\end{proof}

Finally, we have the following theorem: 

\begin{theorem}[MISE of the diffusion estimator]\label{corollary: MISE of g, Rf}
Let $D$ be a bounded domain, $f\in C(\overline{D})$, and $\sigma^{2}(x)$
and $b(x)$ be the variance and bias of the diffusion estimator $g(x;t)$.
Here, if we choose the bandwidth, $\sqrt{t_{N}}$, to satisfy $t_{N}\rightarrow0$,
$Nt_{N}^{d/2}\rightarrow\infty\ (N\rightarrow\infty)$, as $N\rightarrow\infty$,
we have 
\[
\MISE(g(\cdot;t_{N}))\rightarrow0.
\]
Further, when $R(f)<\infty$, as $N\rightarrow\infty$, we have 
\[
\MISE(g(\cdot;t_{N}))=\frac{1}{N(2\sqrt{\pi t_{N}})^{d}}+\frac{t_{N}^{2}}{4}R(f)+o\left(\frac{1}{Nt_{N}^{d/2}}+t_{N}^{2}\right).
\]
\end{theorem}

The first half of Theorem~\ref{corollary: MISE of g, Rf} guarantees that the MISE of the diffusion estimator converges to zero with a weak assumption that $f$ is continuous.
The second half informs the order of the MISE regarding sample size.
Define the asymptotic MISE (AMISE) as
\begin{equation}
    \AMISE (g(\cdot;t)) = \frac{1}{N(2\sqrt{\pi t})^d} + \frac{t^2}{4} R(f).
\end{equation}
The $t$ that minimizes the AMISE is
\begin{equation}
    t^* = \left( \frac{d}{N (2\sqrt{\pi})^d R(f)} \right)^\frac{2}{d+4}, \label{eq: optimal bandwidth}
\end{equation}
and the AMISE under $t^*$ is 
\begin{equation}
    \AMISE (g(\cdot;t^*)) = \frac{d+4}{4} \left( \frac{R(f)}{16 \pi^2 d} \right)^\frac{d}{d+4} \left(\frac{1}{N}\right)^\frac{4}{d+4}.
\end{equation}
This means that, given an appropriate selection of bandwidth, the MISE of the diffusion estimator decreases at the order of $N^{-\frac{4}{d+4}}$.

\subsection{Boundary Consistency}

We also evaluate the point-wise mean squared error (PMSE).
Thus, for the density estimator, $\hat f$, we evaluate 
\begin{equation}
    \PMSE(\hat f(x)) = \Exp |\hat f(x) - f(x)|^2
\end{equation}
by fixing $x\in\overline D$.
Note that the PMSE is definable for $x\in \partial D$ because the diffusion estimator can estimate points on the boundary.

The following theorem guarantees that the diffusion estimator, $g(x;t)$, converges to all points in the distribution, $f(x)$.

\begin{theorem}\label{proposition: MSE of diffusion estimator}
Let $f\in C(\overline D)$. If we select the bandwidth, $\sqrt{t_N}$, to satisfy $\lim_{N\rightarrow \infty}t_N = 0$ and $\lim_{N\rightarrow \infty}Nt_N^{d/2} = \infty$, we have $\PMSE(g(x;t_N))\rightarrow 0$, uniformly, for any $x\in \overline D$.
Further, we have $g(x;t_N)\rightarrow_p f(x)\ (N\rightarrow \infty)$.
\end{theorem}

\begin{proof}
Since the MSE can be decomposed as $\PMSE (g(x,t)) = b^2(x) + \sigma^2(x)$, we show that $b^2(x)$ and $\sigma^2(x)$ both uniformly converge to $0$.
From $\Exp g(x,t) = f_t(x)$, we have $b^2(x) = (f_t(x) - f(x))^2$, and since $f_t(x)$ uniformly converges to $f(x)$ \citep{ito1992diffusion}, $b^2(x)$ uniformly converges to $0$.
Further, from the property of the fundamental solution, there exists $C>0$ that does not depend on $x$, and we have $\sigma^2(x) \leq \frac{1}{N} \Exp_Y |\kappa(x;Y;t)|^2 \leq \frac{C}{Nt^\frac{d}{2}} e^{Ct}$ \citep{ito1992diffusion}.
This uniformly converges to $0$ as $t\rightarrow 0$ and $Nt^\frac{d}{2}\rightarrow \infty$.
\end{proof}

A point to note is that $g(x;t)$ converges to $f(x)$ for all points, including points on the boundary.

\section{Proofs}

\subsection{Proof of Lemma 1}\label{SM:g_M, L infty bound}

Using the triangular inequality, for any data, $\{X_i\}_{i=1}^N$, we have
\begin{align*}
    \|g(\cdot;t) - g_M(\cdot;t)\|_{L^\infty(\overline D)} &= \left\| \sum_{m=M+1}^\infty e^{-\frac{1}{2}\lambda_m t} \left( \frac{1}{N} \sum_{i=1}^N u_m(X_i) \right) u_m \right\|_{L^\infty(\overline D)}\\
    &\leq \sum_{m=M+1}^\infty e^{-\frac{1}{2}\lambda_m t} \left| \frac{1}{N} \sum_{i=1}^N u_m(X_i) \right| \left\| u_m \right\|_{L^\infty(\overline D)} \\
    &\leq \sum_{m=M+1}^\infty e^{-\frac{1}{2}\lambda_m t} \left\| u_m \right\|_{L^\infty(\overline D)}^2.
\end{align*}
Here, there exists constants, $C_1, C_2>0$, that only depend on the domain, $D$, and since $\|u_m\|_{L^\infty(\overline D)}\leq C_1 \lambda_m^\frac{d-1}{4}$ \citep{grieser2002uniform} and Weyl's law, $\lambda_m \sim C_2m^\frac{2}{d}$ holds, we have
\begin{align*}
    \sum_{m=M+1}^\infty e^{-\frac{1}{2}\lambda_m t} \left\| u_m \right\|_{L^\infty(\overline D)}^2 &\lesssim \sum_{m=M+1}^\infty e^{-\frac{1}{2}\lambda_m t} \lambda_m^\frac{d-1}{2} \\
    & \lesssim \sum_{m=M+1}^\infty e^{- C m^{2/d} t} m^\frac{d-1}{d}.
\end{align*}
Note that $C>0$ is a constant that is, for any $m$, $\lambda_m \geq 2 C m^\frac{2}{d}$.

Here, we assume, $M^\frac{2}{d}t\geq \frac{d-1}{2C}$.
The derivative function of $a(z) = e^{- C z^{2/d} t} z^\frac{d-1}{d}$ regarding $z$ is
\begin{equation*}
    \frac{d}{dz} a(z) = \frac{1}{d} \, e^{- C z^{2/d} t} z^{-\frac{1}{d}} \left( d-1-2Cz^\frac{2}{d}t \right),
\end{equation*}
thus $a(z)$ is monotonically decreasing with $z\geq M$.
Thus, because $\sum_{m=M+1}^\infty a(m) \leq \int_M^\infty a(z) dz$, we have
\begin{equation}
    \|g(\cdot;t) - g_M(\cdot;t)\|_{L^\infty(\overline D)} \lesssim \int_M^\infty e^{- C z^{2/d} t} z^\frac{d-1}{d} dz .\label{eq: g_M, int bound}
\end{equation}
Next, we evaluate the integral in the right-hand side of \eqref{eq: g_M, int bound}.
If we variable transform as $w = Cz^{\frac{2}{d}}t$, we have
\begin{equation*}
    \int_M^\infty e^{- C z^{2/d} t} z^\frac{d-1}{d} dz = \frac{d}{2} (Ct)^{\frac{1}{2}-d} \int_{CM^\frac{2}{d}t}^\infty e^{-w} w^{d-\frac{3}{2}} dw.
\end{equation*}
Here, we assume that $M^\frac{2}{d}t$ is large enough to satisfy $e^\frac{CM^{2/d}t}{2} \geq (CM^\frac{2}{d}t)^{d-\frac{3}{2}}$.
Then, for $w \geq CM^\frac{2}{d}t$, we have $e^\frac{w}{2}\geq w^{d-\frac{3}{2}}$, and the following
\begin{equation*}
    \int_{CM^\frac{2}{d}t}^\infty e^{-w} w^{d-\frac{3}{2}} dw \leq \int_{CM^\frac{2}{d}t}^\infty e^{-\frac{w}{2}} dw = 2e^{-\frac{CM^{2/d}t}{2}}
\end{equation*}
holds, and we have
\begin{equation*}
    \int_M^\infty e^{- C z^{2/d} t} z^\frac{d-1}{d} dz \lesssim t^{\frac{1}{2}-d} e^{-\frac{C}{2}M^{2/d}t}.
\end{equation*}
With \eqref{eq: g_M, int bound}, we have
\begin{equation*}
    \|g(x;t) - g_M(x;t)\|_{L^\infty(\overline D)} \lesssim t^{\frac{1}{2}-d} e^{-\frac{C}{2}M^{2/d}t}.
\end{equation*}

\subsection{Proof of Theorem 2}\label{proof: MISE of g_M}

From $t_N\rightarrow 0$ and
\[
    \frac{M_N^{2/d}t_N}{\log(t_N^{-1})}\longrightarrow\infty,
\]
we have
\[
    \exp\left\{\left(\frac{1}{2}-d\right)\log t_N
    -\frac{C}{2}M_N^{2/d}t_N\right\}\longrightarrow 0.
\]
Hence, by Lemma~\ref{lemma: g_M, L infty bound},
\[
    \|g(\cdot;t_N)-g_{M_N}(\cdot;t_N)\|_{L^\infty(D)}
    \longrightarrow 0.
\]
Therefore,
\begin{align*}
    \Exp\|g(\cdot;t_N)-g_{M_N}(\cdot;t_N)\|_{L^2(D)}^2
    &\leq |D|\,
    \|g(\cdot;t_N)-g_{M_N}(\cdot;t_N)\|_{L^\infty(D)}^2 \\
    &\longrightarrow 0.
\end{align*}
Further, Theorem~\ref{corollary: MISE of g, Rf} gives
\begin{align*}
    \MISE(g_{M_N}(\cdot;t_N))
    &\lesssim \MISE(g(\cdot;t_N)) \\
    &\quad + \Exp\|g(\cdot;t_N)-g_{M_N}(\cdot;t_N)\|_{L^2(D)}^2 \\
    &\longrightarrow 0.
\end{align*}

\subsection{Proof of Theorem 3}\label{proof: g_M, weak consistency}

The PMSE can be bounded as
\begin{align*}
    \operatorname{PMSE}(g_{M_N}(x;t_N))
    &\lesssim \Exp|g_{M_N}(x;t_N)-g(x;t_N)|^2
    +\Exp|g(x;t_N)-f(x)|^2 \\
    &\leq \|g_{M_N}(\cdot;t_N)-g(\cdot;t_N)\|_{L^\infty(D)}^2
    +\operatorname{PMSE}(g(x;t_N)).
\end{align*}
Following the argument in Appendix~\ref{proof: MISE of g_M},
\[
    \|g_{M_N}(\cdot;t_N)-g(\cdot;t_N)\|_{L^\infty(D)}
    \longrightarrow 0.
\]
Furthermore, since $\operatorname{PMSE}(g(x;t_N))$ converges uniformly to $0$ by Theorem~\ref{proposition: MSE of diffusion estimator}, the theorem follows.

\bibliographystyle{imsart-nameyear}
\bibliography{ref}

\end{document}